\documentclass[twocolumn,floatfix,aps,prx,superscriptaddress]{revtex4-2}

\usepackage[dvipsnames]{xcolor}
\usepackage[T1]{fontenc}
\usepackage{graphicx}
\usepackage{float}
\usepackage{dcolumn}
\usepackage{bm}
\usepackage{amsfonts,amssymb,amscd,amsmath,amsthm}
\usepackage{mathrsfs}
\usepackage[ruled,linesnumbered]{algorithm2e}
\usepackage{algpseudocode}
\usepackage{enumerate}
\usepackage{epsfig}
\usepackage{caption}
\usepackage{subfig}

\usepackage{dsfont}
\usepackage{physics}
\usepackage{makecell}
\usepackage{bbm}
\usepackage{svg}
\usepackage{tikz}
\usepackage[colorlinks = true]{hyperref}
\hypersetup{colorlinks=true, linkcolor=blue, citecolor=blue, urlcolor=black}
\usepackage{amsmath, amsthm, amsfonts, amssymb, amsbsy, mathtools, xcolor, bm, bbm, graphicx, changepage, cleveref}

\usepackage{epstopdf}
\usepackage{framed}
\usepackage{multirow}
\usepackage{color}
\usepackage{comment}
\usepackage{qcircuit}
\usepackage{faktor}
\usepackage[normalem]{ulem}

\makeatletter
\newtheorem*{rep@theorem}{\rep@title}
\newcommand{\newreptheorem}[2]{%
\newenvironment{rep#1}[1]{%
 \def\rep@title{#2 \ref{##1}}%
 \begin{rep@theorem}}%
 {\end{rep@theorem}}}
\makeatother

\newtheorem{theorem}{Theorem}

\newtheorem{lemma}{Lemma}

\newtheorem*{assumption*}{Assumption}
\newtheorem{corollary}{Corollary}

\usepackage{lineno}

\DeclareMathOperator{\support}{\mathrm{supp}}

\newcommand{\normH}[1]{{\left\vert\kern-0.25ex\left\vert\kern-0.25ex\left\vert #1
   \right\vert\kern-0.25ex\right\vert\kern-0.25ex\right\vert}}

\newcommand{\mc}{\mathcal}

\begin{document}
\title{Entanglement-Induced Resilience of Quantum Dynamics}
 \author{Tianfeng Feng}
 \address{QICI Quantum Information and Computation Initiative, Department of Computer Science, The University of Hong Kong, Pokfulam Road, Hong Kong SAR, China}
\author{Yue Cao}
\address{QICI Quantum Information and Computation Initiative, Department of Computer Science, The University of Hong Kong, Pokfulam Road, Hong Kong SAR, China}
\author{Wenjun Yu}
\address{QICI Quantum Information and Computation Initiative, Department of Computer Science, The University of Hong Kong, Pokfulam Road, Hong Kong SAR, China}
\author{Junkai Zeng
}\address{Shenzhen International Quantum Academy, Shenzhen 518048, China}
\author{Xiaopeng Li}
\email{xiaopeng_li@fudan.edu.cn}
\address{State Key Laboratory of Surface Physics, Key Laboratory of Micro and Nano Photonic Structures (MOE),
and Department of Physics, Fudan University, Shanghai 200433, China}

\author{Xiu-Hao Deng}
\email{dengxiuhao@iqasz.cn}
\address{Shenzhen International Quantum Academy, Shenzhen 518048, China}
\address{Shenzhen Branch, Hefei National Laboratory, Shenzhen, 518048, China}
\author{Qi Zhao}
\email{zhaoqi@cs.hku.hk}
\address{QICI Quantum Information and Computation Initiative, Department of Computer Science, The University of Hong Kong, Pokfulam Road, Hong Kong SAR, China}
\address{Shenzhen International Quantum Academy, Shenzhen 518048, China}
\date{\today}

\begin{abstract}
Quantum many-body devices suffer from imperfections that destabilize dynamics and limit scalability. We show that the dynamical growth of entanglement can intrinsically protect generic quantum dynamics against coherent and perturbative noise. Through rigorous theoretical analysis of general quantum dynamics and numerical simulations of spin chains and fermionic lattices, we prove that entanglement-entropy growth confines the influence of local Hamiltonian perturbations, thereby suppressing errors in dynamical errors. The degree of protection correlates quantitatively with the entanglement entropy of subsystems on which the perturbations act, and applies broadly to both analog quantum simulators and real-time control protocols. This entanglement-induced resilience is conceptually distinct from quantum error correction or dynamical decoupling: it passively leverages native many-body correlations without additional qubits, measurements, or control overhead. Our results reveal a generic mechanism linking entanglement growth to dynamical stability and provide practical guidelines for designing noise-resilient quantum devices.

\end{abstract}
\maketitle

Quantum many-body systems hold great promise for advancing technologies ranging from analog quantum simulators to fault-tolerant quantum computers \cite{RevModPhys.86.153Nori,preskill2012quantumcomputingentanglementfrontier,buluta2009quantum, QECPhysRevA.57.127}. However, a major obstacle limits their practical realization: all real-world systems suffer from imperfections. Static disorder—such as fabrication-induced variations in qubit parameters—may dramatically alter quantum phases  \cite{MBLRevModPhys.91.021001,AndersonPhysRev.109.1492,heyl2018dynamical,wang2015fidelity}. Equally problematic is crosstalk between components or control pulse distortions, which creates structured errors that resist conventional error suppression methods \cite{ crosstalkPRXQuantum.3.020301, doria2011optimal,altafini2012modeling, schwenk2017estimating}. These imperfections destabilize quantum dynamics, threatening the scalability of quantum devices. Current solutions like error-correcting codes or dynamical decoupling often require significant overhead \cite{fukui2018high,decouplingPhysRevLett.82.2417,PhysRevX.6.041034}. 
This raises a critical question: Can quantum systems naturally exhibit inherent resilience against these coherent or perturbative imperfections without external intervention?

\begin{figure}
\centering
\includegraphics[width=0.43\textwidth]{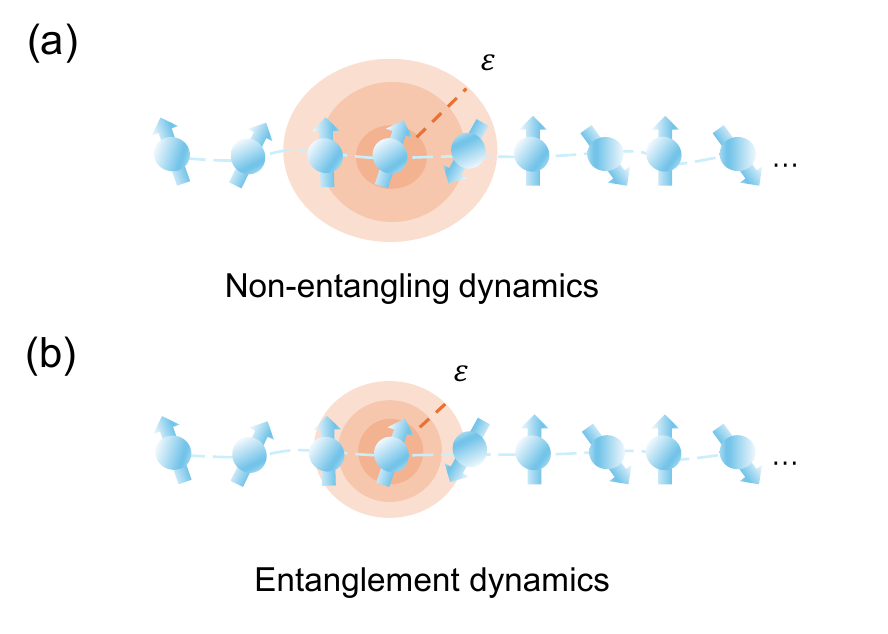}
    \caption{Entanglement-induced resilience of quantum dynamics. The red shaded band shows the error envelope. Consider a local error Hamiltonian acting on a spin chain, where the error terms involve $k$ near-neighbour qubits and comprise  $\text{poly}(k)$ Pauli operators with small coefficients. (a) When the quantum system exhibits weak or no entanglement, the dynamical error may scale linearly with the number of error terms, i.e., $\text{poly}(k)$. (b) In contrast, sufficient entanglement suppresses the dynamical error, leading to a scaling proportional to the square root of the number of error terms, i.e., $\sqrt{\text{poly}(k)}$. }
    \label{fig:entanglement}
\end{figure}

Entanglement is a hallmark of quantum correlations that plays a central role in the dynamics of many-body systems \cite{VVRevModPhys.80.517, EntanglementRevModPhys.81.865, JensRevModPhys.82.277, MBLRevModPhys.91.021001}. It underpins a broad range of tasks in quantum information processing and quantum computing \cite{bennett1993teleporting,jozsa1997entanglement,giovannetti2011advances}. Recent studies suggest that, in certain regimes,  entanglement might play a stabilizing role in simulating quantum systems \cite{Zhao_2025NP,feng2025trotterizationoperatorscramblingentanglement}. In digital quantum simulations \cite{lloydUniversalQuantumSimulators1996,childsTheoryTrotterError2021}, the generation of entanglement helps reduce the algorithmic errors caused by the discretization of time evolution (Trotterization) \cite{Zhao_2025NP}. 
However, digital quantum simulation remains experimentally challenging, and despite extensive efforts \cite{lloydUniversalQuantumSimulators1996,childsTheoryTrotterError2021,lowHamiltonianSimulationQubitization2019,berrySimulatingHamiltonianDynamics2015}, a significant gap persists between current capabilities and the requirements for scalable, fault-tolerant implementations \cite{childsFirstQuantumSimulation2018}. In contrast, analog simulation provides a promising near-term platform not only for strongly coupled fermionic systems such as the Fermi-Hubbard model \cite{arovas2022hubbard,esslinger2010fermi,hart2015observation, shao2024antiferromagnetic} but also for a variety of other many-body models, including spin models \cite{bernien2017probing,simon2011quantum}, bosonic systems \cite{gross2017quantum,yang2020observation}, and lattice gauge theories 
\cite{ wiese2013ultracold,jordan2012quantum,zhou2022thermalization}. 

While most analog quantum simulations have traditionally focused on ground-state properties and the characterization of quantum phases, our research emphasizes the (real-time) dynamical behavior of quantum systems and simulators. A fundamental question arises: does entanglement universally enhance robustness in continuous quantum dynamics—such as in analog simulations or real-time control protocols?
Specifically, does entanglement passively coexist with reliable dynamics, or does it actively suppress intrinsic dynamical errors? This question is especially pertinent in many-body systems where imperfection and disorder intricately interact with complex spatial-temporal correlations.

In this work, we reveal that the dynamic growth of entanglement provides intrinsic protection against perturbations, including disorder and crosstalk couplings of unitary quantum dynamics. A conceptual overview of our results is presented in Fig. \ref{fig:entanglement}, where error accumulation is suppressed through entanglement.
By combining analytical models and numerical simulations (spin chains, fermionic lattices), we show that rapidly increasing entanglement entropy confines local perturbations, preventing their spread. This mechanism applies broadly to analog quantum simulations and beyond, where greater entanglement correlates with higher operational fidelity.
Furthermore, to bridge the gap between theory and experiment, we introduce a dynamical entanglement detection scheme based on error subspaces. We demonstrate that measuring experimentally accessible correlation functions of the noise terms serves as an effective probe to diagnose whether the dynamics is protected by entanglement, circumventing the need for full quantum state tomography.
Our results reshape the understanding of entanglement’s role in realistic dynamics: while static entanglement is often fragile under local, uncorrelated noise, its dynamic generation in many-body evolution can inherently suppress coherent errors—distinct from engineered quantum error correction. We identify this as a natural error suppression mechanism embedded in many-body quantum dynamics. Practically, these insights offer guidelines for designing robust quantum protocols: intentionally engineering entanglement growth through tailored interactions or control sequences can enhance resilience without additional resources such as extra qubits or syndrome measurements.
Furthermore, larger entanglement regions are classically hard to simulate \cite{cirac2021matrix,orus2019tensor}, which not only enhances the quantum advantage but also suggests that extensive entanglement can further bolster error suppression and computational power in many-body systems. These findings immediately benefit analog quantum simulators by leveraging native entanglement dynamics to improve robustness without additional error correction \cite{mcewen2021removing,harrington2025synchronous}.\\


\section{Entanglement induces dynamic resilience}

\textbf{Disorders and imperfections in Quantum dynamics.---}
Errors in quantum dynamics arise from perturbations—imperfections of coupling and disorder—that perturb the real-time evolution of quantum states. These dynamical errors are ubiquitous across platforms, not only in analog quantum simulators but also in digital emulation and real-time control protocols. One can model the perturbed evolution by
   $ H'(t)=H_0(t)+H_{\text{pert}}(t)$, yielding the evolved state 
  $  \ket{\psi'(t)}=\mathcal{T}e^{-i\int H^\prime(t) dt}\ket{\psi(0)}$ with $\mathcal{T}$ the time-ordering operator. 

Generally, perturbations from stochastic disorder and  imperfections  are unified under a generalized perturbative Hamiltonian term
\begin{equation}
   H_{\text{pert}}(t) = \sum_k \delta_k V_k + \sum_m \eta_m N_m,
\end{equation}
where $\delta_k$ parameterizes random disorder (e.g., pulse distortions in quantum control or lattice potential fluctuations in analog simulations), while $\eta_m N_m$ represents time-invariant noise from systematic errors, such as fabrication-induced variations, residual crosstalk couplings in multi-qubit gates \cite{crosstalkPRXQuantum.3.020301} or the effective Trotterization artifacts in digital algorithms \cite{childsTheoryTrotterError2021}). 
Crucially, disorder and imperfection strength bifurcate dynamical responses: Weak perturbation 
induces perturbative energy shifts without altering phase topology, as seen in trapped-ion spectral jitter; Strong disorder or imperfection, 
however, may drive symmetry-breaking or topological phase transitions---MBL shifts entanglement scaling from volume-law to area-law in spin chains \cite{MBLRevModPhys.91.021001}.


\begin{figure*}
    \centering
    \subfloat[]{
    \includegraphics[width=0.42\linewidth]{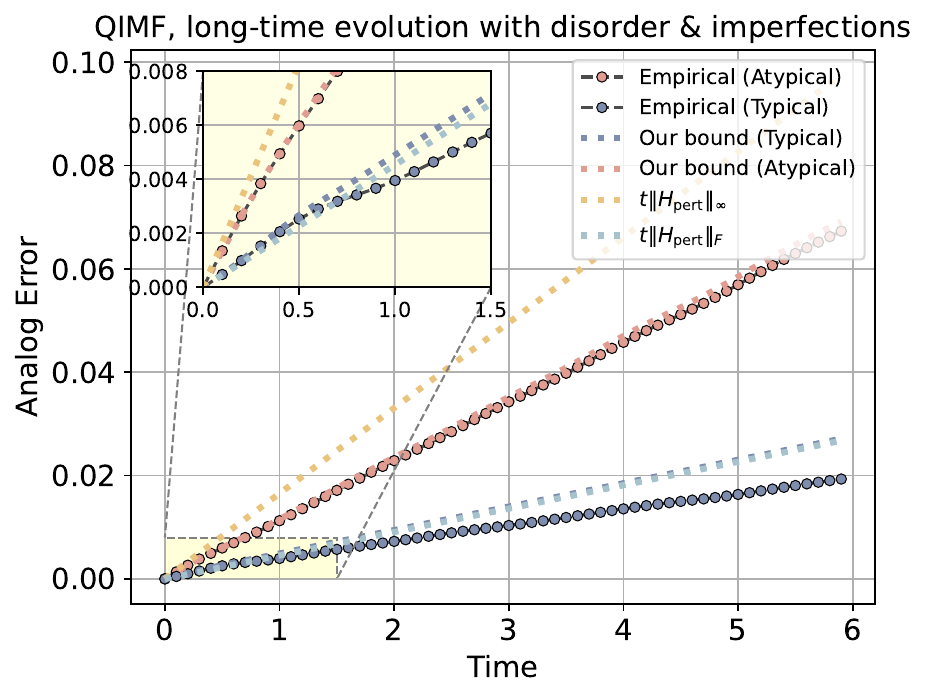}
    }
    \subfloat[]{
    \includegraphics[width=0.47\linewidth]{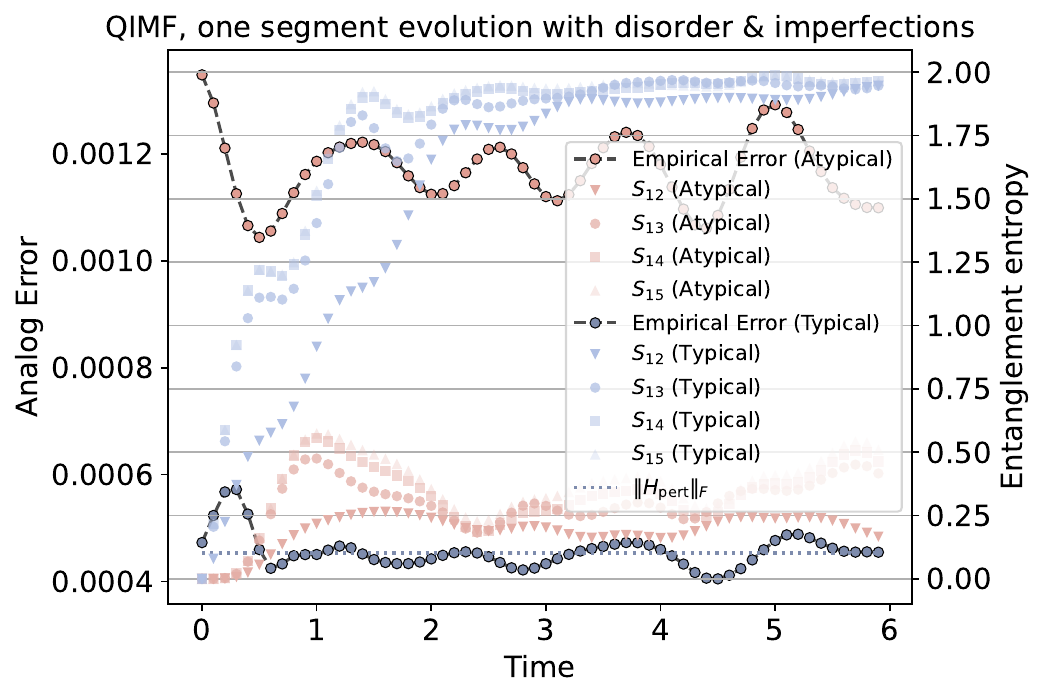}
    }
    \caption{Error of Analog quantum simulation of 1D QIMF model in typical and atypical cases. The perturbation Hamiltonian is comprised of disorder terms with coefficients $\delta_i\in\mathcal N(0,0.01)$ and an imperfection term $\eta=0.01$. (a) Long-time analog simulation error. In typical input state, where the entanglement entropy grows during the system's evolution, our estimate closely matches the average-case performance, scaling with $\norm{H_{\text{pert}}}_F$. Notably, at early times when entanglement has not yet developed sufficiently, the error curve exhibits a slight curvature with a slope higher than the average-case estimate. However, after the entanglement saturates (around $t\approx0.5$), the error slope aligns with the average case, and our bound and the average-case bound run approximately parallel. In contrast, in an atypical input state, where the entanglement entropy remains low throughout the evolution, the error scaling deviates significantly from the average case and approaches the worst-case bound characterized by the spectral norm. 
(b) One-segment simulation error, $  \|(U_0(\delta t)-U(\delta t))\ket{\psi(t)}\|_F$, evaluated for different initial states with $\delta t=0.1$. The corresponding entanglement entropy of two qubits ($S_{1,2},~S_{1,3},~S_{1,4},~S_{1,5}$) is also shown, illustrating how entropy growth correlates with error suppression. }
    \label{fig:mainplaceholder}
\end{figure*}

\textbf{Entanglement-induced resilience of quantum dynamics.---} 
Without loss of generality, we begin with the time-dependent scenario, where the time-independent case can be treated as a special simplification.  
Suppose $H_0(t)$ is the ideal Hamiltonian and $H^\prime (t)$ is the noisy Hamiltonian with perturbation, the evolution of the ideal and noisy unitaries are given as $U_0(t)=\mathcal{T} e^{-i\int_{s=0}^{t}d{s}H_0(s)}$ and  $U(t)=\mathcal{T} e^{-i\int_{s=0}^{t}d{s}H^\prime(s)}$ respectively.
Here, we would like to quantify the distance between $U_0(t)$ and $U(t)$ with initial state $\ket{\psi(0)}$. 
To explore the role of entanglement in the quantum dynamics, 
we utilize the vector norm $\norm{(U_0(t)-U(t))\ket{\psi(0)}}$ to track the quantum dynamics errors (Here the vector norm is defined as $\norm{X\ket{\psi}}=\sqrt{\bra{\psi}X^\dagger X\ket{\psi}}$, and the spectral norm $\norm{X}=\text{max}_{\ket{\psi}}\norm{X\ket{\psi}}$, capturing the worst input state case).  
It is shown that $W(t)=U(t)^\dagger U_0(t)$ is the unique solution of $ i \partial t W(t) =U (t)^\dagger (H(t)-H^\prime(t))U(t ) W(t)$ with $W(0)=I$ \cite{HHKL8555119,doi:10.1126/science.1121541}, by unitary invariant property of vector norm and the triangle inequality, we can bound the dynamic error as (see section A in Materials and Methods)
\begin{equation}
\begin{split}
\label{maineq2}
    \norm{(U_0(t)-U(t))\ket{\psi}} \le\int_0 ^t d\tau \norm{\big[  H^\prime (\tau)-H_0(\tau)\big]\ket{\psi(\tau)}}.
        \end{split}
\end{equation}
This upper bound on the error is governed by the Hamiltonian mismatch $H^\prime(\tau)-H_0(\tau)$, which in general arises for time‑dependent Hamiltonians. Here $\ket{\psi(\tau)}$ denotes the ideally evolved state at time $\tau$. 
Since the perturbed Hamiltonian can be expressed as
$H^\prime(\tau)=H_0(\tau)+H_{\text{pert}}(t)$, 
The upper bound of the dynamic error now is given as $\int_0 ^t d\tau \norm{ H_{\text{pert}}(t) \ket{\psi(\tau)}}$, i.e.,
\begin{equation}
\label{Eq:maineq}
\begin{split}
         \int_0 ^t d\tau \sqrt{\bra{\psi(\tau)} H_{\text{pert}}^\dagger (t)H_{\text{pert}}(t) \ket{\psi(\tau)}}.\\
\end{split}
\end{equation}
This bound can characterize errors in various dynamical processes, including real-time quantum control. 

To better characterize this bound, we connect it to the entanglement entropy of the quantum state. For simplicity, here we consider the perturbation $H_{\text{pert}}$
to be time-independent 
and consists of local perturbative terms $H_{\text{pert}}=\sum_j {H_{\text{pert}}}_j$
(We put the time-dependent analysis in Methods.)
It is shown that the expectation value of a positive semi-definite operator can be related to the entanglement entropy of a quantum state ~\cite{Zhao_2025NP}.
Specifically, let $A=\sum_j A_j$ be a positive semi-defined operator, acting on $N$ qubits, where $A_j$ acts nontrivially on the subsystem $\support(A_j)$, then
 $ | \bra{\psi}A \ket{\psi}|\le \frac{\Tr(A)}{d}+  \Delta_{A(\psi)}, $
where $\Delta_{A(\psi)} =\sum_{j} ~\| A_j\| \sqrt{2\log(d_{\support(A_{j})})-2S(\rho_{j})},$
$\rho_{j}:=\Tr_{[N]\setminus \support(A_j)}(\ket{\psi}\bra{\psi})$ is the reduced density matrix of $\ket{\psi}\bra{\psi}$ on the subsystem of $\support(A_j)$, and $\text{S}(\rho_{j})$ is the entanglement entropy of $\rho_{j}$ ~\cite{Zhao_2025NP}. Since $H_{\text{pert}}^\dagger H_{\text{pert}}=\sum_{j,j^\prime}{H_{\text{pert}}^\dagger}_j{H_{\text{pert}} }_{j^\prime} $ is a positive operator, the bound of error of the quantum dynamics is now given by the degree of entanglement entropy of various partitions, e.g., $ t\norm{H_{\text{pert}}}_F+\int_0 ^t d\tau \sqrt{\Delta_{{H^\dagger_{\text{pert}}}{H_{\text{pert}}}} (\psi(\tau))}$,
where $\norm{X}_F=\sqrt{\Tr(X^\dagger X)/d}$, where $d$ is the dimension of the operator $X$ (see Supplementary Section A).



Generally, quantum states tend to be thermalized after evolution, i.e., the entropy of the subsystem grows universally. 
Suppose $t=c$, state entanglement is growing to case: for $t\ge c$, $\Delta_{H_{\text{pert}}}(\psi(\tau))\approx0$ so that $\norm{H_{\text{pert}} \ket{\psi(\tau)}}\approx \norm{H_{\text{pert}} }_F$. In this situation,
we have

\begin{equation}
\begin{split}
         \norm{(U_0(t)-U(t))\ket{\psi}}&\lesssim 
                 \int_0 ^c d\tau \sqrt{\bra{\psi(\tau)} H_{\text{pert}} ^\dagger  H_{\text{pert}}  \ket{\psi(\tau)}}\\
                 &\quad+(t-c) \norm{H_{\text{pert}} }_F.\\
\end{split}\label{Eq:mainlongtime}
\end{equation}
Here $c$ can be efficiently estimated by measuring the expectation value of $\bra{\psi(t)}H_{\text{pert}}^\dagger H_{\text{pert}}\ket{\psi(t)}$. Typically, $c$ is a small constant; as shown in Fig.~\ref{fig:mainplaceholder}, we find $c\approx 0.5$ in the 1D quantum Ising model with a transverse field with specified parameters (for more numerical details, see the analog quantum simulation section).
When $t\ge c$, the error scaling exhibits average-case behavior, analogous to that for random input states, such as Haar-random states or states forming a unitary 1-design, e.g. 
           $\int_{\psi\in \text{Harr}} d\psi \norm{ (U_0(t)-U(t))\ket{\psi}} \le   t \norm{H_{\text{pert}} }_F$ (see  Supplementary Section A).
Therefore, our entanglement-induced bound achieves the average-case scaling when $t\ge c$.

In general, the perturbation Hamiltonian $H_{\mathrm{pert}}$ whose error scales at most polynomially with the system size $N$, i.e., $\norm{H_{\mathrm{pert}}} \propto N^{p}$ (with $p=1$ for linear scaling).
Our results show that entanglement confers intrinsic robustness against dynamical errors, so that error bounds are expressed in the normalized Frobenius norm $\norm{H_{\mathrm{pert}}}_F$ instead of the spectral norm $\norm{H_{\mathrm{pert}}}$; this, in turn, affords tighter and more reliable analyses of quantum dynamics since $\norm{H_{\mathrm{pert}}}_F\le\norm{H_{\mathrm{pert}}}$. For example, if $H_{\mathrm{pert}}$ comprises $N$ bounded error terms, then one generally has $\norm{H_{\mathrm{pert}}} \propto N $ and $\norm{H_{\mathrm{pert}}}_{F} \propto \sqrt{N}$. This indicates that our entanglement-based analysis suggests a potential quadratic improvement over the worst-case scenario, as the normalized Frobenius norm typically captures average performance \cite{zhaoHamiltonianSimulationRandom2021,chen2024average, kahanamokumeyer2025logdepthinplacequantumfourier}.

Besides, the statistical properties of disorder $\sum_k\delta_k V_k$ may also cancel some errors in quantum dynamics \cite{https://doi.org/10.4230/lipics.tqc.2024.2}.
Suppose $\delta_k$ is a random parameter of  Gaussian distribution with $\mathbb{E} \delta _k=0$ and $\text{Var}(\delta_k)=\sigma^2_k$. By considering the nature of $\delta_k$, the density matrix of the real evolved state is given as 
$\rho(t)=  \mathbb{E}_{\{\delta_1,\delta_2,...\}}U (t)\ket{\psi}\bra{\psi}U^\dagger(t)$.
We find and numerically test that the upper bound of the trace distance $\norm{ \rho_0(t)-\rho(t) }_1$ is determined by the variance of disorder and the imperfection of coupling, which mainly contribute to the error of the evolved density matrix, where entanglement only has an inhibitory effect on the imperfection crosstalk terms (See section C in Materials and Methods).

\section{Analog quantum simulation}

In this section, we focus on analog quantum simulation \cite{buluta2009quantum,farhiAnalogAnalogueDigital1998,daley2022practical}. Building on the general analysis of quantum dynamics presented above, we numerically demonstrate that entanglement plays a pivotal role in suppressing simulation errors.
We first test the one-dimensional (1D) quantum Ising model with a transverse field (QIMF). Specifically, the Hamiltonian of 1D QIMF is given as, 

\begin{equation}
    H_0=h_x\sum_{i=1}^NX_i+h_y\sum_{i=1}^NY_i+J\sum_iX_iX_{i+1},
    \label{eq:QIMF}
\end{equation}
 where the parameters $h_x=0.809,h_y=0.9045,J=1$ \cite{PRXQuantum.4.010311}. 
We assume that the empirical Hamiltonian is $H=H_0+H_{\text{pert}}$.
The source of perturbation $H_{\text{pert}}$ includes stochastic disorder $H_{\text{dis}}=\sum_k\delta_k V_k$ and
intrinsic imperfections $H_{\text{imp}}=\sum_m \eta_m N_m$, i.e., $H_{\text{pert}}=H_{\text{dis}}+H_{\text{imp}}$. 
In our numerical tests, the disorder term $H_{\text{dis}}$ is set as a combination of $N$ single-qubit Pauli operators $X_i$'s, i.e., $H_{\text{dis}}=\sum_{i=1}^N\delta_iX_i$, where $\delta_i$ is randomly chosen with respect to a normal distribution and $X_i$ represents Pauli $X$ acting on the $i$-th qubit. The imperfection term is set as $H_{\text{imp}}=\eta\sum_{i=1}^{N-1}X_iX_{i+1}$, where $\eta$ is a constant.

We compare the simulation error 
for a typical initial state versus an atypical initial state. 
Fig. \ref{fig:mainplaceholder} (a) demonstrates the tightness of our theoretical result for long-time analog simulations. In typical scenarios ($\ket{\psi(0)}=\ket0^{\otimes N}$), where the entanglement entropy grows during the system's evolution, our estimate closely matches the average-case performance, scaling with $\norm{H_{\text{pert}}}_{F}$. In contrast, in atypical cases ($\ket{\psi(0)}=\ket+^{\otimes N}$), where entanglement remains low, the error scaling deviates significantly and is closer to the worst-case spectral norm bound $\norm{H_{\text{pert}}}$. 
Fig. \ref{fig:mainplaceholder} (b) further illustrates how entropy growth correlates with error suppression: rapid entropy increase corresponds to error reduction near the average estimate, while low entropy is associated with larger errors. This comparison highlights the role of entanglement in suppressing analog simulation errors, as captured by our upper bounds. Additional validation in a 2D QIMF model is provided in Supplementary Section C.


\begin{figure*}
    \centering
    \subfloat[]{
        \includegraphics[width=0.45\linewidth]{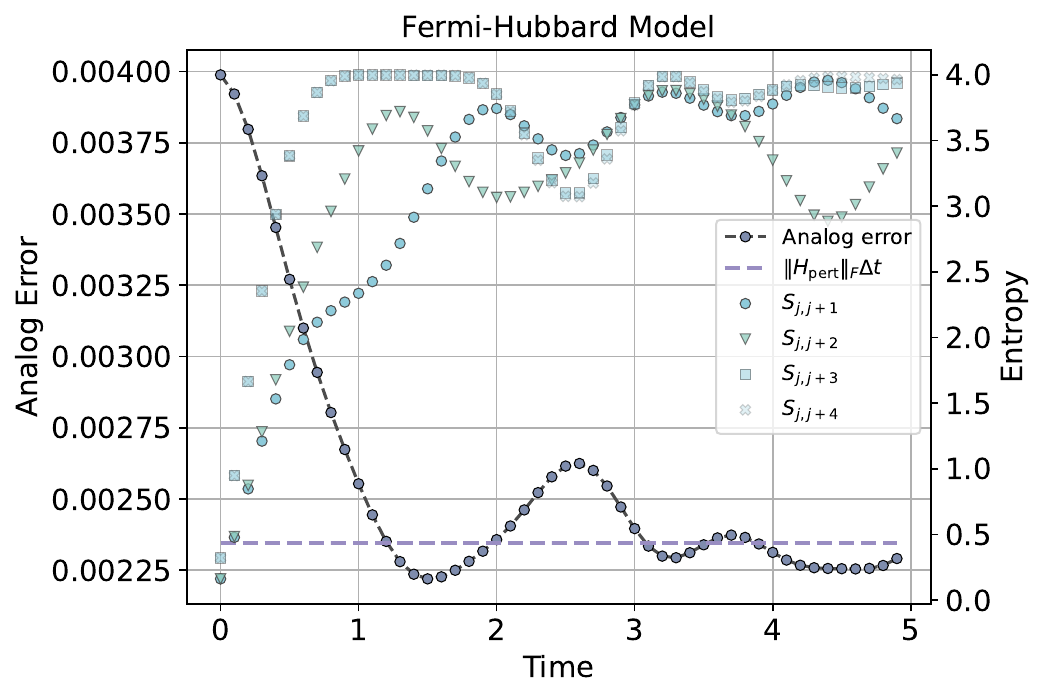}}
    \subfloat[]{
    \includegraphics[width=0.45\linewidth]{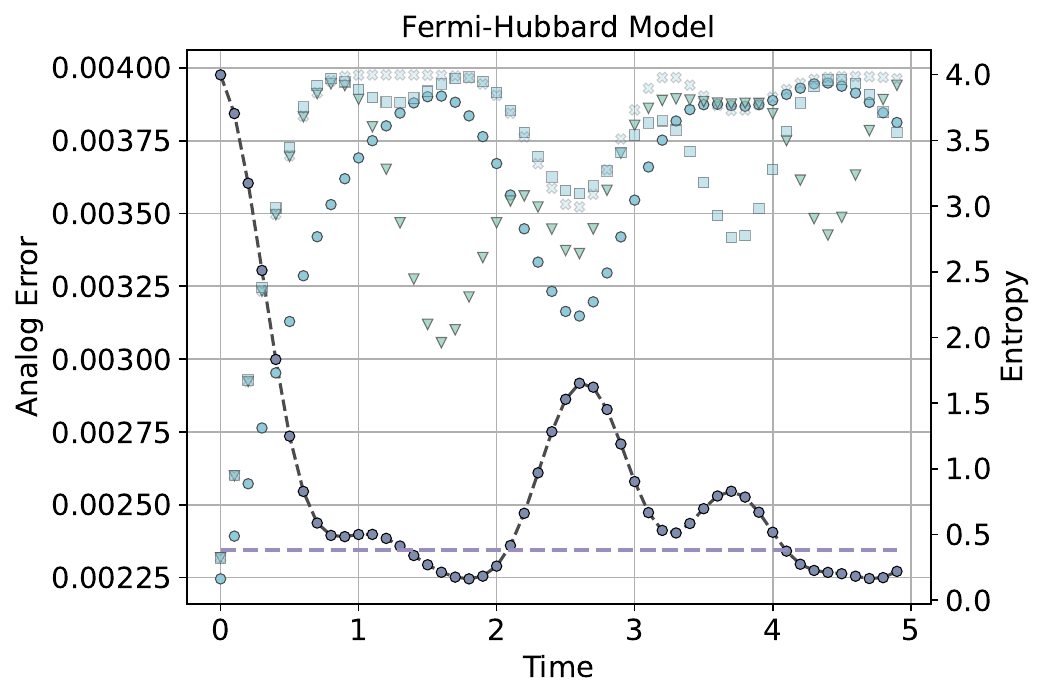}}
    
    \subfloat[]{
    \includegraphics[width=0.45\linewidth]{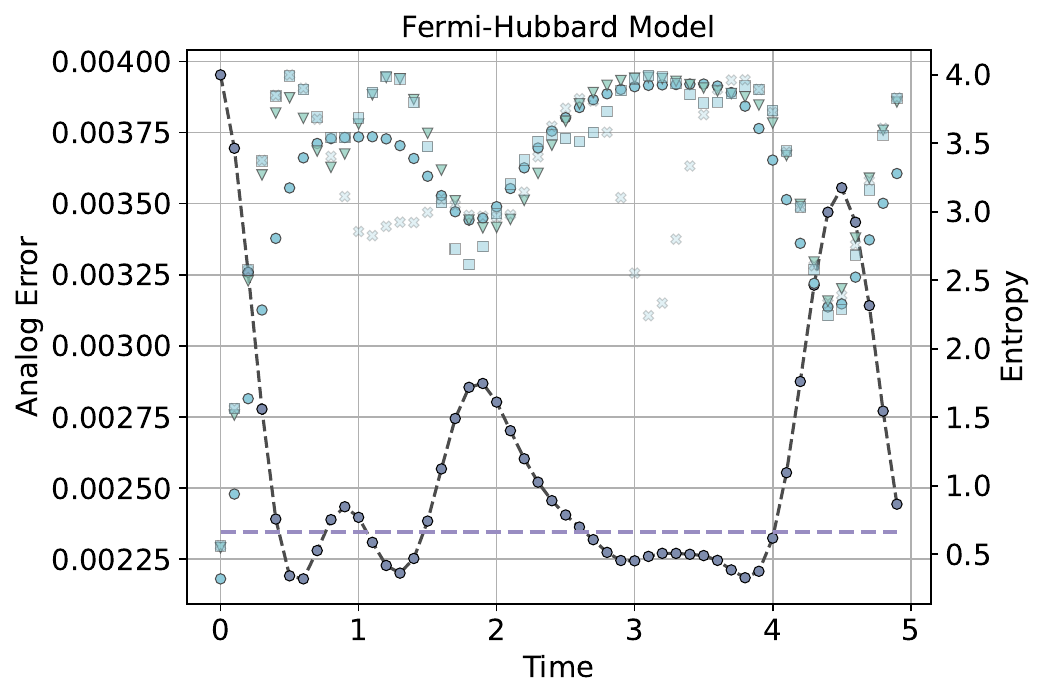}}
    \subfloat[]{
    \includegraphics[width=0.45\linewidth]{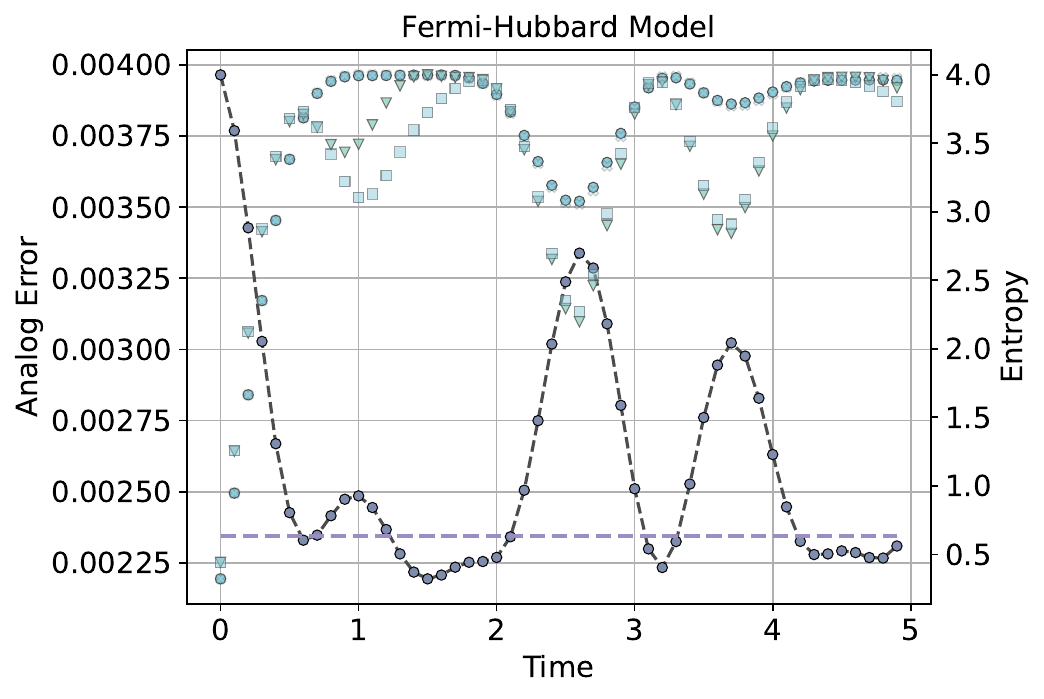}
    }

    \caption{Error of one-segment analog simulation with Fermi-Hubbard model. 
    The vertical axis represents the corresponding empirical error for each individual segment with segment length $\Delta t=0.1$. 
    The system scale is set as $L=8$ and the noise strength is $\delta=0.01$.  The initial state is prepared as (a) $\ket{\uparrow\downarrow}_1\ket{\uparrow\downarrow}_2\ket{\uparrow\downarrow}_3\ket{\uparrow\downarrow}_4$, (b) $\ket{\uparrow\downarrow}_1\ket{\uparrow\downarrow}_2\ket{\uparrow\downarrow}_4\ket{\uparrow\downarrow}_5$, (c) $\ket{\uparrow\downarrow}_1\ket{\uparrow\downarrow}_3\ket{\uparrow\downarrow}_5\ket{\uparrow\downarrow}_7$, (d) $\ket{\uparrow\downarrow}_1\ket{\uparrow\downarrow}_3\ket{\uparrow\downarrow}_4\ket{\uparrow\downarrow}_6$. The entanglement entropy is defined on the subsystem of two sites numbered as $(j,j+1),~(j,j+2),~(j,j+3),~(j,j+4)$, respectively. Entanglement entropy and the simulation error display an inverse relationship: as the entanglement entropy increases to the maximum 4 during the evolution, the simulation error reduces to the Frobenius norm level; when the entanglement entropy drops, the error exhibits sharp spikes.}
    \label{mainfig:fermi}
\end{figure*}

Another case study is the Fermi-Hubbard model, which is fundamental in condensed matter physics, used to describe interacting fermions on a lattice \cite{arovas2022hubbard,esslinger2010fermi,hart2015observation, shao2024antiferromagnetic}. It captures the competition between kinetic energy and on-site Coulomb repulsion, providing insights into phenomena such as magnetism, metal-insulator transitions, and high-temperature superconductivity. We consider a 1D Fermi-Hubbard chain with periodic boundary conditions, governed by the Hamiltonian
\begin{equation}
    H_0=-\sum_{i=1}^L\sum_{\sigma}(a_{i,\sigma}^\dagger a_{i+1,\sigma}+\text{H.c.})+V\sum_{i=1}^La_{i,\uparrow}^\dagger a_{i,\uparrow}a_{i,\downarrow}^\dagger a_{i,\downarrow},
\end{equation}
where $\sigma=\uparrow, \downarrow$, $V$ denotes the on-site Coulomb repulsion, $a_{i,\sigma},a_{i,\sigma}^\dagger$ represent annihilation and creation operators, respectively.
These operators are characterized by the anticommutation expressed as
\begin{equation}
    \{a_i,a_j^\dagger\}=\delta_{ij},\quad \{a_i,a_j\}=\{a_i^\dagger,a_j^\dagger\}=0.
\end{equation}
The dynamics of Fermi-Hubbard model keeps the particle number $N=\sum_i(a^\dagger_{i,\uparrow}a_{i,\uparrow}+a^\dagger_{i,\downarrow}a_{i,\downarrow})$ and the z-component of the total spin $S_z=\sum_i(a^\dagger_{i,\uparrow}a_{i,\uparrow}-a^\dagger_{i,\downarrow}a_{i,\downarrow})$ invariant, with periodic boundary condition ensures the lattice translational symmetry.

In our numerical simulations, we introduce an additional Coulomb perturbation (error term) to the Hamiltonian, i.e.,
\begin{equation}
H_{\text{pert}}=\delta\sum_{i=1}^La_{i,\uparrow}^\dagger a_{i,\uparrow}a_{i,\downarrow}^\dagger a_{i,\downarrow}.
\label{FH}
\end{equation}
To analyze the error at each time step, we discretize the evolution into multiple short segments, each with a duration of $\Delta t=0.1$. The one-segment simulation error is defined as $\norm{U(\Delta t)-U_0(\Delta t)\ket{\psi(t)}}$. Here the initial state is set as $\otimes_{i=1}^4\ket{\uparrow\downarrow}_{i_1}$, which produces the maximal simulation error with noise model $H_{\text{pert}}$ and particle number $N=8$. 
The entanglement entropy is defined on subsystems with two sites (see Materials and Methods for more details). 
Fig.~\ref{mainfig:fermi} depicts the relationship between one-segment simulation errors and entanglement entropy. A clear inverse correlation is observed: as entanglement entropy increases, the dynamical error decreases; conversely, as entanglement entropy decreases, the error increases. This trend highlights the inverse correlation between entanglement growth and dynamical error, which is perfectly match our  theoretical analysis. We present the numerical results for the one-segment error at longer times in Fig. \ref{fig:longtimeFH} (Materials and Methods).

\section{Dynamical entanglement detection with error subspaces}

Our results show that, although the error bound is formally captured by measurements of $H_{\text{pert}}^\dagger H_{\text{pert}}$, the unknown coefficients of individual noise terms make such measurements impractical. What is experimentally accessible, however, is the \emph{structure} of the noise---the Pauli operators dictated by the device architecture.

This suggest measuring the expectations of these structure terms—akin to correlation functions—probes entanglement across specified partitions. For any traceless local Pauli operator $P$, the expectation $\langle\psi|P|\psi\rangle$ approaches zero when the reduced state on its support becomes nearly maximally mixed, indicating strong entanglement. The operator $H_{\text{pert}}^\dagger H_{\text{pert}}$ contains cross terms
${H_{\text{pert},j}^\dagger H_{\text{pert},j'}}/{\|H_{\text{pert},j}\|\,\|H_{\text{pert},j'}\|}$
with coefficient–independent weights, making their expectation values experimentally accessible correlation probes.
Specifically, if for all pairs $j \neq j'$,
$
\langle \psi(t)| H_{\text{pert},j}^\dagger H_{\text{pert},j'} + \text{H.c.} |\psi(t) \rangle/{\|H_{\text{pert},j}\|\,\|H_{\text{pert},j'}\|} \approx 0
$, the error behavior is Frobenius-norm–like, consistent with substantial entanglement within the supports $\support{H_{\text{pert},j}}$. By contrast, significant deviations from zero suggest limited entanglement, and the error may be close to the spectral norm case.


Figure~\ref{fig:entangle_det} illustrates these ideas in a disordered 1D QIMF with control imperfections (the same model as in the first analog setting). For typical initial states at $t=6$, more than 170 pairs of 
cross-term operators exhibit expectations near zero,
indicating substantial entanglement generated during the evolution and consistent with Frobenius-norm–like error behavior. In contrast, atypical initial states yield many nonzero expectations, indicating limited entanglement growth. These results show that measuring such correlation functions provides an effective diagnostic of dynamical entanglement across the partitions of interest \cite{Hauke_2016}.

Note that this criterion provides a necessary—but not sufficient—condition for strong entanglement; conversely, a significant nonzero value is sufficient to rule out high entanglement (a separability witness).
To reliably confirm the entanglement, we recommend measuring additional observables—at least two Pauli operators (e.g., $XX$ and $ZZ$), whose joint statistics provide complementary constraints on the structure of the state and yield a more robust indication of genuine entanglement.

\begin{figure}
    \centering
    \includegraphics[width=0.85\linewidth]{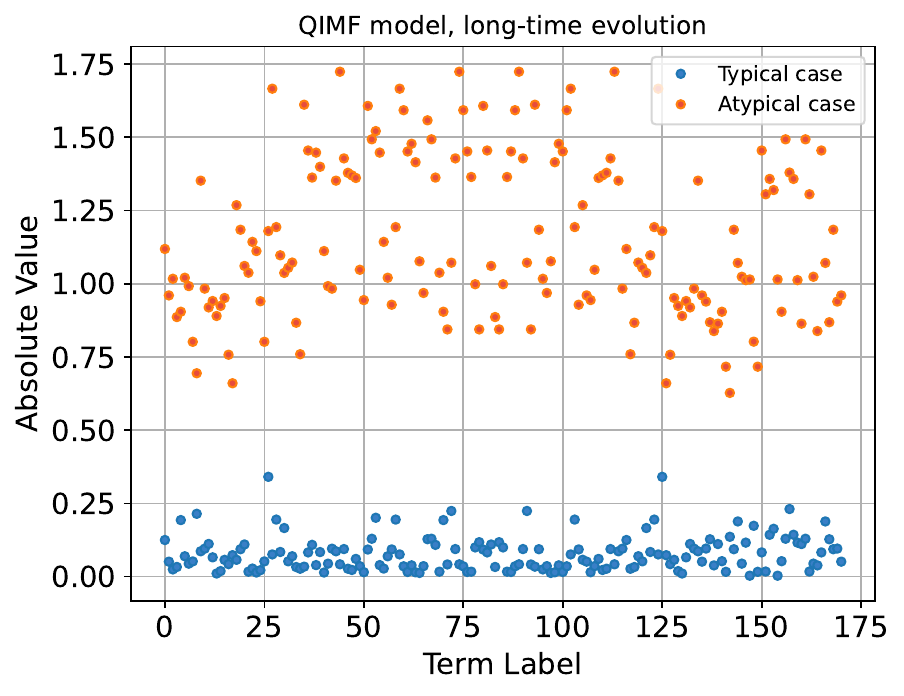}
    \caption{The contribution of Hermitian cross terms for both the typical and the atypical cases. The noise terms include both disorder and control imperfections. The Hermitian contribution is defined as ${\bra{\psi(t)} {{H^\dagger}_{\text{pert} j}} H_{\text{pert}j^\prime}+\text{H.C.}\ket{\psi(t)}}/{\norm{{{H^\dagger}_{\text{pert} j}} } \norm{H_{\text{pert}j^\prime}}}$, where $\ket{\psi(t=6)}$ is the evolved state for typical and atypical cases and $H_{\text{pert}j^\prime}$ is a Pauli string component in the perturbation Hamiltonian. These observables track whether the error behaves close to the Frobenius norm (typical) or drifts toward spectral‑norm (atypical) scaling.}
    \label{fig:entangle_det}
\end{figure}

\section{Quantum control protocols}

Another important application of our analysis lies in quantum control—the design and implementation of protocols that manipulate quantum systems to realize quantum gates with high fidelity. 

We consider implementing a $k$-qubit gate on an $n$-qubit register. In general, the $k$ target qubits are coupled to $m$ neighboring spectator qubits (often referred to as spectators or intruders~\cite{deng2021correcting}); together, these $k + m$ qubits form a subsystem of the full $n$-qubit system.
We denote by $H^{m,k}(t)$ the Hamiltonian with quantum control acting on the $m+k$-qubit (targets + neighbors), and its evolution is $U(t)$. The ideal gate $U_0(t)$ is generated by the ideal Hamiltonian $H_{\text{ideal}}(t)$ on the target subspace.

For an initial $n$-qubit state $\ket{\psi}$, the effective error of the quantum gate can be given as $\norm{(I_{n\setminus k}\otimes U_0(t)-I_{n\setminus (m+k)}\otimes U(t))\ket{\psi}}$, which we hereafter simplify to $\norm{(U_0(t)-U(t))\ket{\psi}}$ by omitting the identity operators.
To illustrate the effects of entanglement in quantum control in a simple and clear manner, we partition the entire system into two subsystems, $A$ and $B$.
Let $A$ denote the system of $m+k$ qubits, and $B$ denote the system of the remaining qubits $n-m-k$, and now 
$\ket{\psi}=\ket{\psi}_{AB}$ (as shown in Fig. \ref{mainfig:QD_2D} (a)). 
By utilizing our entanglement-based result of time-dependent analysis (see Materials and Methods), one may have
    $\norm{ (U_0(t)_A-U(t)_A)\ket{\psi}_{AB}} \approx \int_0^t d\tau \norm{(H_{\text{ideal}}(\tau)-H^{m,k}(\tau) }_F$ for large entangled state $\ket{\psi}_{AB}$. 
This suggests that entanglement can induce robustness in implementing quantum gates with quantum control protocols.\\

\begin{figure*}
    \centering
    
    \subfloat[]{\includegraphics[width=0.30\textwidth]{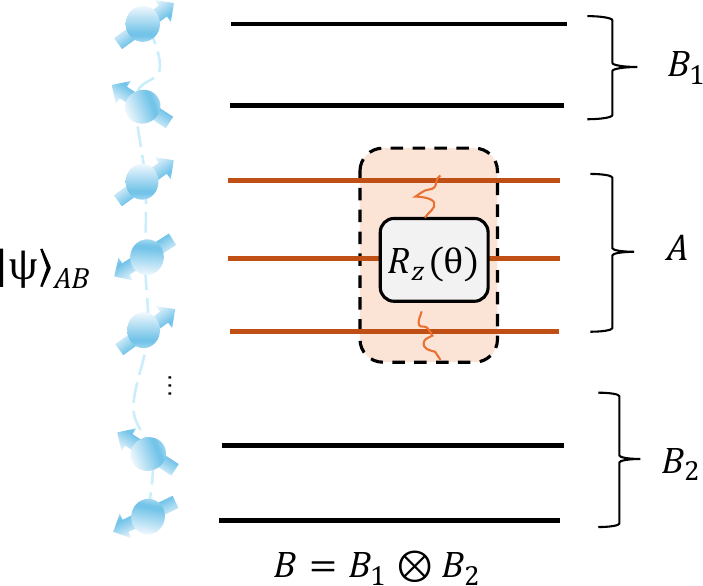}}
    \subfloat[]{\includegraphics[width=0.33\textwidth]{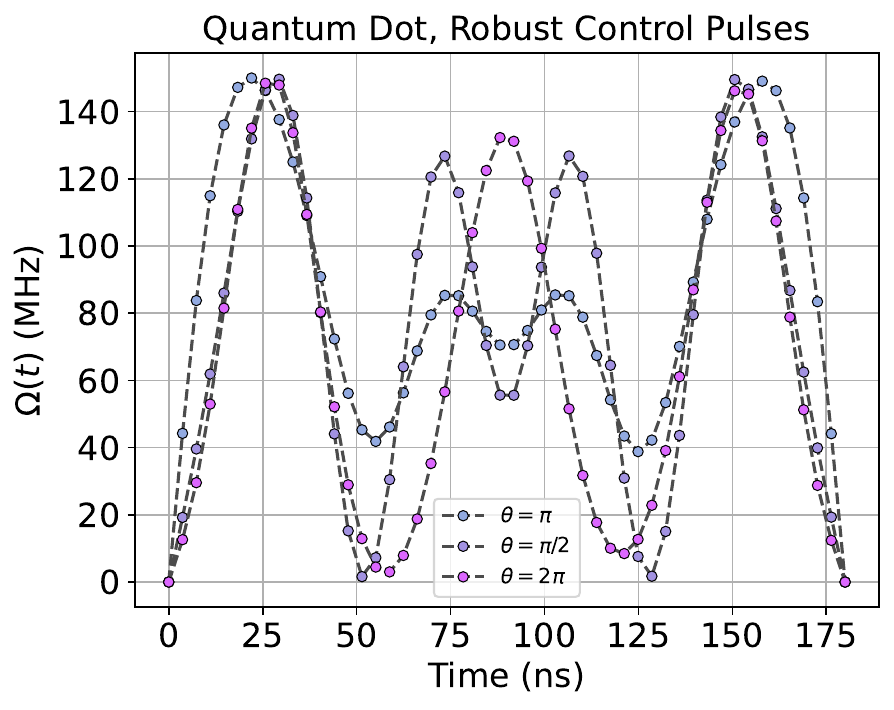}}
    \subfloat[]{\includegraphics[width=0.33\textwidth]{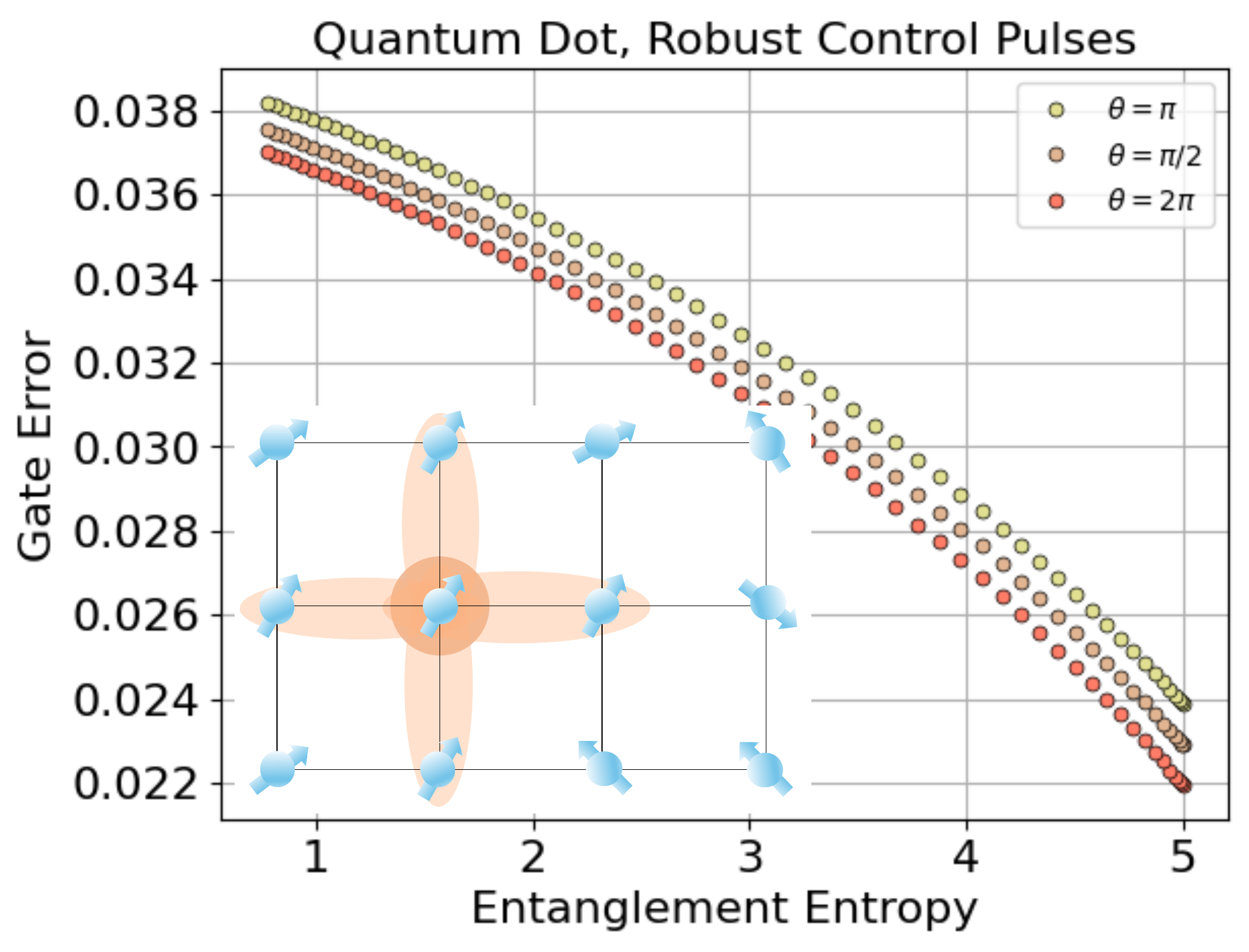}}\
    \caption{Subsystem entanglement and gate error. (a) Schematic partition of a quantum-control circuit for a target single-qubit gate into subsystem $A$ (the target qubit plus nearby qubits directly coupled to the control fields) and subsystem $B$ (the remaining register qubits).  (b) Quantum-dot system driven with robust control pulses for a single-qubit gate. The optimal driving waveform for implement time $T=180$ns and target gates $X_\pi,X_{\pi/2},X_{2\pi}$, respectively. (c) Dependence of effective gate error on the entanglement entropy (higher temperature of purified Gibbs state, higher entanglement entropy) of the initial states with corresponding waveforms. We consider a 2D lattice in which 4 spectator qubits couple to the target (so $|A| = 5$ and $S_{\text{max}} = \log(2^5) = 5$).  States are sampled as Gibbs states with temperatures $T$ chosen as $T\in\left\{\frac1t|t\in\{0.008,0.016,\dots,0.48\}\right\}$ (equivalently $T$ spanning $\approx 2.08 ~\text{to}~125$).  
    As the entanglement entropy $S(A)$ approaches its maximum, the error bound rapidly decreases, indicating that greater preexisting entanglement between $A$ and $B$ suppresses the influence of gate imperfections.
   }
    \label{mainfig:QD_2D}
\end{figure*}


To demonstrate how entanglement enhances quantum control of quantum gates, we perform numerical simulations of quantum gate implementation on a $3\times4$ qubit 2D quantum dot system (see the lattice in Fig. \ref{mainfig:QD_2D} (c)). Specifically, to realize the Pauli gate $U(\theta)=\exp(-iX_t\theta)$ on the target $k$ qubit (here $k=1$, and the target qubit is surrounded by orange circle in Fig. \ref{mainfig:QD_2D} (c)), one drives the system with an ideal control field that generates the single-qubit Hamiltonian $H_{\text{ideal}}=\dfrac{\Omega(t)}2X_k$. Nevertheless, imperfections in the inter-qubit couplings, together with (random) drift of the control and frequency, generate a total time-dependent perturbation term $H_{\text{pert}}(t)=H_{\text{drift}}+H^{m,k}-H_{\text{ideal}}$ in the Hamiltonian. 
The perturbation Hamiltonian is then expressed accordingly 
 (see Supplementary section B)
\begin{equation}
    \begin{aligned}
        H_{\text{pert}}(t) &= \delta \sum_{i\in R_k}Z_i + \epsilon\Omega_k(t)X_k \\&+ \frac{J}{4} \sum_{i\in R_k}Z_iZ_k
         + \frac{\tan\theta \Omega_k(t)}{2} \\ 
         & \times\sum_{i\in R_k}
         \Big[\cos(\Delta \tilde{E}_z t) X_iZ_k - \sin(\Delta \tilde{E}_z t) Y_iZ_k \Big],
    \end{aligned}
    \label{main_Eq_multi_noise}
\end{equation}
where $R_k$ represents the target qubit and its adjacent qubits as illustrated by orange ovals in Fig. \ref{mainfig:QD_2D} (c) where $k=1$ and $m=4$ ($|R_k|=4$).  
For best control accuracy, we can set the driving field $\Omega(t)$ as a robust control pulse (RCP) \cite{hai2025scalablerobustquantumcontrol}.
We test the performance of RCPs for 3 different target gates $X_{\pi},X_{\pi/2},X_{2\pi}$. The gate time is $T=180$ns, and the waveform of these RCPs are demonstrated in Fig.~\ref{mainfig:QD_2D} (b).

To show the relationship between entanglement entropy and the faithfulness of the implementation of quantum gates, we consider a set of two-partite purified Gibbs states or thermofield double states \cite{cottrell2019build}, 
\begin{equation}
    \ket{\psi_T}_{AB}=\dfrac{\sum_ie^{-\epsilon_i/T}\ket{\phi_i}_{R_k}\ket{\phi^*_i}_{B_1}}{\|\sum_ie^{-\epsilon_i/T}\ket{\phi_i}_{R_k}\ket{\phi^*_i}_{B_1}\|}\otimes\ket 0^{\otimes 2}_{B_2},\label{maineq:Gibbs}
\end{equation}
where the system is divided into $R_k$ and the complementary part. $B_1$ consists of 5 qubits responding to $R_k$, while $B_2$ represents the other 2 qubits. Here $\{(\epsilon_i,\phi_i)\}$'s are eigenvalues and eigenvectors of the 5-qubit Hamiltonian $V=-\sum_{i\in R_k}Z_i$. 
The entanglement entropy of $\ket{\psi_T}$ increases as the temperature $T$ increases. 
Fig.~\ref{mainfig:QD_2D}(c) shows the relationship between the entanglement entropy and the gate error, with parameters $\Delta E_Z=200$MHz, $J=100$kHz, $\delta =50$kHz, $\epsilon=0.001$ and $\theta=\arctan\dfrac J{2\Delta E}$. The error quickly reduces when the entanglement rises close to its maximum. This suggests that higher entanglement entropy is correlated with improved fidelity, indicating that entanglement can further enhance the robustness of quantum control protocols even with RCPs already applied. Additional validation of a class of two-qubit gates in a 2D quantum dot system is provided in Supplementary Section C.

\section{Discussions and conclusions}

Our results demonstrate that the dynamical growth of many-body entanglement can provide intrinsic protection against the coherent and perturbative noises of quantum dynamics. In the context of the Eigenstate Thermalization Hypothesis \cite{ETHPhysRevA.43.2046,ETHPhysRevE.50.888}, this implies a fundamental connection between error proliferation and the entanglement and thermalization of the system—a connection that persists even in regimes of incomplete thermalization. Moreover, extending these ideas to non-Hermitian quantum dynamics (Parity-Time symmetry system \cite{PTPhysRevLett.80.5243, PT2PhysRevLett.98.040403}) indicates that entanglement may still serve as a vital resource in mitigating errors. Recent studies have shown that entanglement can even protect against noise in cross-platform protected qubits \cite{qjhp-8x6zlogicalqubitE} and operator dynamic simulation \cite{feng2025trotterizationoperatorscramblingentanglement}, suggesting that harnessing entanglement could be a powerful strategy in developing noise-resilient quantum technologies.

Translating these insights into the language of quantum circuits, 
we define the error unitary as the product of the ideal circuit $U_0$ and the approximate circuit $U$ (see section D in Supplementary Materials). 
Our entanglement-based analysis of quantum circuits suggest that, given a fixed initial state, entanglement growth in a quantum circuit decreases the sensitivity of parameter changes in the sequential circuit so which may induce phenomena such as barren plateaus \cite{ortiz2021entanglement,PhysRevResearch.3.033090}, hindering the efficiency of variational algorithms. While rapid entanglement generation can impede gradient-based optimization, from a state preparation perspective \cite{gleinig2021efficient,Zhang_2022}, it enhances robustness against noise. This dual role underscores the importance of balancing entanglement control: suppressing excessive growth to facilitate optimization, while leveraging its benefits to improve system stability.

In general,  harnessing entanglement provides a quantitative framework for understanding and enhancing the stability of complex quantum systems, thereby facilitating the development of noise-resilient quantum technologies. On a fundamental level, these insights may connect to phenomena such as quantum phase transitions \cite{osterloh2002scaling,EntanglementQFTPhysRevA.66.032110,HaiqingLin2004PhysRevLett.93.086402} and the emergence of robust properties or classical behavior from quantum many-body systems \cite{zurek2022emergence}.  Understanding these connections deepens our comprehension of the fundamental physics and could inspire new strategies for designing robust quantum devices capable of operating in noisy environments.

\section*{MATERIALS and Methods}

\subsection{Universal error Upper bound for time-dependent Hamiltonian}
The derivation of our main results will be based on the following lemma concerning a time-dependent Hamiltonian.
\begin{lemma}     \label{HHKL}\cite{HHKL8555119}
    Let $A(t)$ and $B(t)$ be continuous time-dependent Hermitian operators, and let $U_A(t)$ and $U_B(t)$  with $U_A(0) = U_B(0) = I$ be the corresponding time evolution unitaries. 
    Then $W(t)=U_B(t)^\dagger U_A(t)$ is the unique solution of $ i \partial t W(t) =U_B (t)^\dagger (A(t)-B(t))U_B(t ) W(t)$, and $W(0)=I$.
\end{lemma}

Suppose $A(t)=H_0(t)$ is the ideal Hamiltonian and $B(t)=H^\prime (t)$ which is the noise Hamiltonian, we have $W(t)=U_{B}(t)^\dagger U_A(t)=\mathcal{T}e^{\int_{s_1=t}^{0}ds_1iH^\prime(s_1)} \mathcal{T}e^{-\int_{s_2=0}^{t}d{s_2}iH_0(s_2)}$.  Define $\mathcal{D}=\norm{(U_A(t)-U_B(t))\ket{\psi}}$.
By making use of Lemma \autoref{HHKL}, we bound the dynamical error as

\begin{equation}
\begin{split}
\mathcal{D}
        &=\norm{U_B(t)( W(t)-I)\ket{\psi}}\\
       &=\norm{( W(t)-I)\ket{\psi}}\\
       &=\norm{( W(t)- W(0))\ket{\psi}}\\
        &=\norm{i\int^t_0 dt  \partial \tau W(\tau)\ket{\psi}}\\
        &= \norm{\big[\int_0 ^t d\tau U_B (\tau)^\dagger (A(\tau)-B(\tau))U_B(\tau) W(\tau) \big]\ket{\psi}}\\
           &= \norm{\big[\int_0 ^t d\tau U_B (\tau)^\dagger (A(\tau)-B(\tau))U_A \big]\ket{\psi}}\\
        &\le \int_0 ^t d\tau\norm{\big[  (A(\tau)-B(\tau))U_A(\tau)\big]\ket{\psi}}\\
        &=\int_0 ^t d\tau \norm{\big[  H^\prime (\tau)-H_0(\tau)\big]e^{-\int_s^{\tau} d siH(s)}\ket{\psi}}\\
        &=\int_0 ^t d\tau \norm{\big[  H^\prime (\tau)-H_0(\tau)\big]\ket{\psi(\tau)}}.
        \end{split}
\end{equation}
In the derivation above, we have used the norm-preserving property of the unitary operator: $\|U X\ket{\psi}\| = \|X\ket{\psi}\|$ for any operator $X$. Our result suggest, for a time-dependent Hamiltonian, the upper bound of error term in the state vector reduces to $H'(\tau) - H_0(\tau)$, where $\ket{\psi(\tau)}$ denotes the ideal time-evolved state at evolution time $\tau$.

\subsection{Entanglement-based bound of time-dependent Hamiltonian}

In general, the perturbative Hamiltonian, $H_{\text{pert}}=H^\prime (\tau)-H_0(\tau)$, may be time-dependent, e.g, $H_{\text{pert}}(\tau)$. To investigate the role of entanglement in bounding the error for such time-dependent cases, we discretize the total evolution time $t$ into $J$ segments of duration $\Delta t= t/J$.  The upper bound $\int_0 ^t d\tau \norm{ H_{\text{pert}}(\tau)\ket{\psi(\tau)}}$ can then be rewritten as a sum over these segments: 
\begin{equation}
   \sum_{j=0}^{J-1} \int_{j\Delta t} ^{(j+1)\Delta t} d\tau \sqrt{\bra{\psi(\tau)} H_{\text{pert}}^\dagger(\tau)  H_{\text{pert}}(\tau) \ket{\psi(\tau)}}. \label{eqtimed}
\end{equation}
A key mathematical relationship is employed. For a operator $A=\sum_j A_j$ composed of non-trivial local operators  $A_j$, the following inequality holds~\cite{Zhao_2025NP} 
\begin{equation}
   \bra{\psi}A^\dagger A\ket{\psi}\le \norm{A}_F^2+\Delta_{A^\dagger A}(\psi),
\end{equation}
where $\Delta_{A^\dagger A(\psi)} =\sum_{j} ~\| A_j\| \sqrt{2\log(d_{\support(A_{j})})-2S(\rho_{j})},$
$\| \cdot\|_F$
denotes the normalized Frobenius norm, $\rho_{j}:=\Tr_{[N]\setminus \support(A_j)}(\ket{\psi}\bra{\psi})$ is the reduced density matrix of $\ket{\psi}\bra{\psi}$ on the subsystem of $\support(A_j)$, and $\text{S}(\rho_{j})$ is the entanglement entropy of $\rho_{j}$. 
Defining $A^{(j+1)\Delta t}_{j\Delta t}=\int_{j\Delta t} ^{(j+1)\Delta t} d\tau H_{\text{pert}}^\dagger(\tau)  H_{\text{pert}}(\tau) $ as a positive operator for segment $j$, we bound the dynamical error as (see Supplemental Section A for details)
\begin{equation}
\begin{split}
& \norm{(U_A(t)-U_B(t))\ket{\psi(0)}}
  \le (\Delta t)^{\frac{1}{2}} \sum_{j=0}^{J-1} \sqrt{\frac{\Tr(A^{(j+1)\Delta t}_{j\Delta t})}{d}} \\& \quad +
   \sum_{j=0}^{J-1}\int_{j\Delta t} ^{(j+1)\Delta t} d\tau  \sqrt{\Delta_{{H^\dagger_{\text{pert}}(\tau)H_{\text{pert}}(\tau)}} (\psi(\tau))  }.
   \end{split}
   \label{timed_methods}
\end{equation}
Assuming the state $\ket{\psi(\tau)}$ remains sufficiently entangled at all times, the upper bound approximately simplifies to
\begin{equation}
 \sum_{j=0}^{J-1} (\Delta t)^{\frac{1}{2}}\sqrt{\frac{\Tr(A^{(j+1)\Delta t}_{j\Delta t})}{d}}=\int_{0}^t \norm{H_{\text{pert}}(\tau)}_F.
    \label{eq:t-dependent_entangled}
\end{equation}
Finally, it can be verified that in the time-independent limit where $H_{\text{pert}}(\tau)=H_{\text{pert}}$, we have $A^{(j+1)\Delta t}_{j\Delta t}=\Delta t H^\dagger_{\text{pert}}H_{\text{pert}} $,
In this case, the general upper bound in \eqref{timed_methods} reduces to 
\begin{equation}
     t\norm{H_{\text{pert}}}_F+\int_{0} ^{t} d\tau  \sqrt{\Delta_{{H^\dagger_{\text{pert}}(\tau)H_{\text{pert}}(\tau)}} (\psi(\tau))},
\end{equation}
which recovers the exact upper bound for the long-time error in the time-independent case in the main text.




\subsection{Error analysis with random disorder}

Suppose that $ H_{\text{pert}} = V+N=\sum_k \delta_k V_k + \sum_m \eta_m N_m$ and suppose $\delta_k$ is a random parameter of  Gaussian distribution with $\mathbb{E} \delta _k=0$ and $\text{Var}(\delta_k)=\sigma^2_k$.  The upper bounds of Eq. \eqref{Eq:maineq}  and Eq. \eqref{Eq:mainlongtime} are for a single run with fixed $\delta_k$. Considering the random nature of $\delta_k$, the density matrix of the real evolved state is given as 
$\rho(t)=  \mathbb{E}_{\{\delta_1,\delta_2,...\}}U (t)\ket{\psi}\bra{\psi}U^\dagger(t)$, where 
$U(t)=\hat{T}e^{-i\int H^\prime(t) dt}$.
Now, the distance between the ideal evolution and the approximated one should be defined by the 1-norm $\norm{\rho_0(t)-\rho(t)}_1$ ($\norm{A}_1=\Tr(\sqrt{A^\dagger A})$) with $\rho_0(t)=U_0(t)\ket{\psi}\bra{\psi}U^\dagger_0(t) $.
Since for any two projectors $\ket{\chi}\bra{\chi}$ and $\ket{\phi}\bra{\phi}$, the following inequality holds, e.g., $\norm{\ket{\chi}\bra{\chi}-\ket{\phi}\bra\phi{}}_1\le 
2 \norm{\ket{\psi}-\ket{\phi}}$ \cite{watrous2018theory}.
Now we have (see Supplementary Section A)

\begin{equation}
\begin{split}
\norm{ \rho_0(t)-\rho(t) }_1 
&\le 2 \int_0 ^t d\tau \sqrt{\sum_k \sigma^2_k+\norm{ N\ket{\psi(\tau)}}^2},\\
&\le 2t \sqrt{\sum_k \sigma^2_k} +2\int_0 ^t d\tau \norm{ N\ket{\psi(\tau)}},
\end{split}
\end{equation}
where $N=\sum_m \eta_m N_m$. 
This suggests that the variance of disorder and the imperfection of coupling mainly contribute to the error of the evolved density matrix, where entanglement only has an inhibitory effect on the imperfection term $N$.
Note that the average case error $ \mathbb{E}_{\{\delta_1,\delta_2,...\}} \norm{(U_0(t)-U(t))\ket{\psi}} =\frac{1}{2}\norm{ \rho_0(t)-\rho(t) }_1$, while $\norm{ \rho_0(t)-\rho(t) }_1$ provide the rigorous upper bound of the error. 

\begin{figure}
    \centering
    \includegraphics[width=0.9\linewidth]{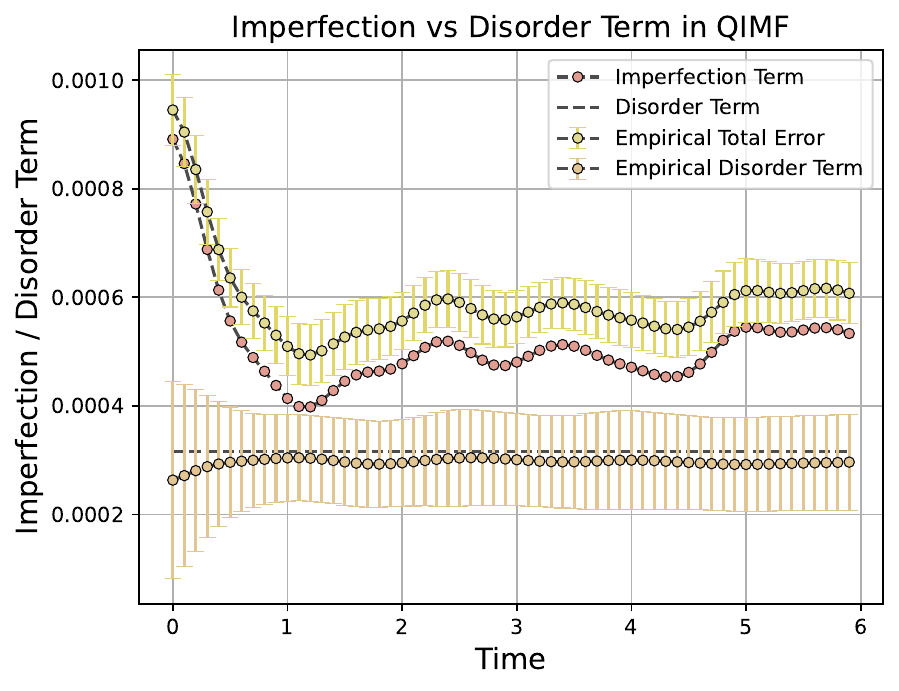}
    \caption{Comparison between the contribution to analog simulation error of disorder and imperfection terms. Here we use the same model as that used in Figure~\ref{fig:mainplaceholder}. We take 30 trials for each moment, with disorder strength $\delta_k$ randomly drawn with respect to the normal distribution $\mathcal N(0,0.01)$. The error bar in the figure shows the standard deviation for the 30 trials. The contribution of disorder and imperfections is defined as $\int_t^{t+\Delta t}d\tau\|V\ket{\psi(\tau)}\|$ and $\int_t^{t+\Delta t}d\tau\|N\ket{\psi(\tau)}\|$, respectively, where $\Delta t$ represents the segment length. The result shows that while the average disorder contribution is almost independent of the state, the contribution of imperfections quickly converges to the entangled case as the evolution time increases.}
    \label{fig:imp_vs_dis}
\end{figure}

Figure~\ref{fig:imp_vs_dis} compares the contribution of imperfection terms and disorder terms on the analog simulation error. It shows that the average contribution of disorder terms remains at almost the same level as the system evolves, while the imperfection terms depends more on the state, which is consistent with our theoretical results.

Remarkably, if the Gaussian distribution of disorder is unbiased, i.e. $\mathbb{E} \delta _k\ne0$, there are some cross terms consisting of the disorder terms and the imperfection terms. The error is also relate to the disorder term but entanglement can also contributes to suppress these error.

\subsection{Fermionic lattices}

\begin{figure*}
    \centering
    \subfloat[]{\includegraphics[width=0.45\textwidth]{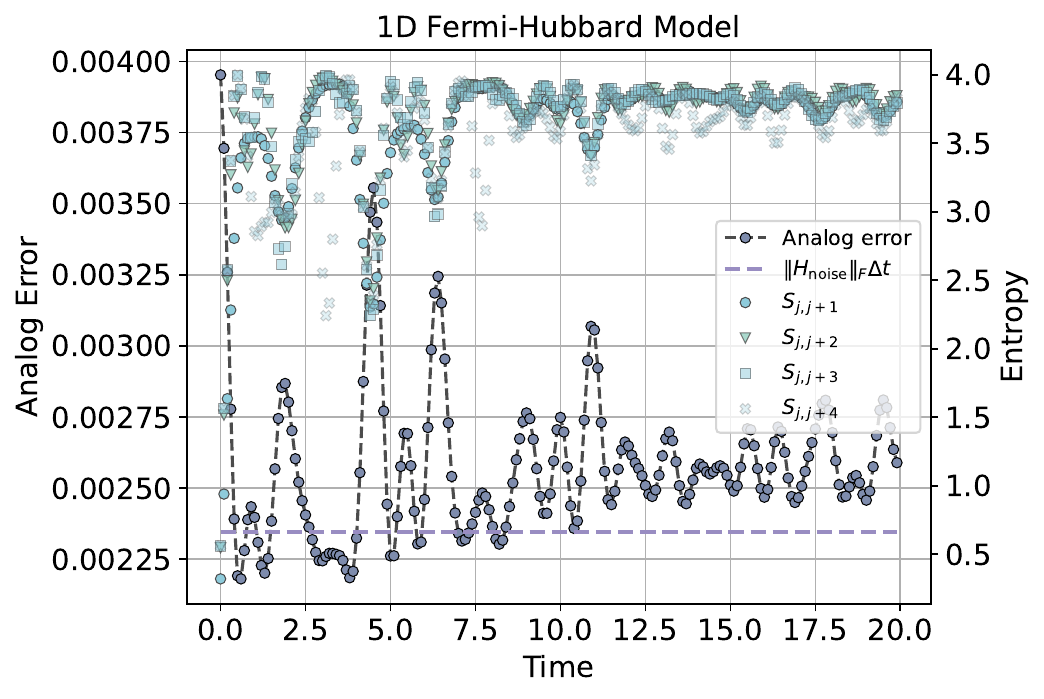}}
    \subfloat[]{\includegraphics[width=0.45\textwidth]{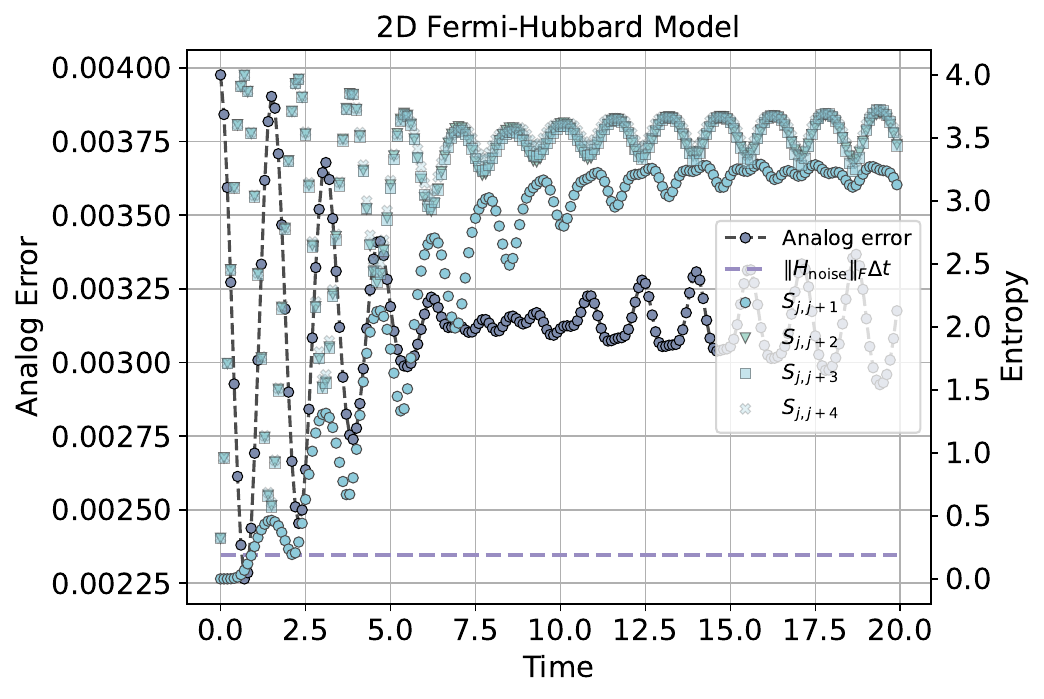}}
    \caption{Error of one-segment analog simulation with (a) 1D Fermi-Hubbard model with 8 sites and (b) $2\times4$ 2D Fermi-Hubbard model in a longer evolution process. The Coulomb interaction strength is $U=0.5$, the segment length is set as $\Delta t = 0.1$, and the noise strength is $\delta = 0.01$. Periodic boundary conditions are applied for both models. The initial state is prepared as $\ket{\uparrow\downarrow}_1\ket{\uparrow\downarrow}_3\ket{\uparrow\downarrow}_5\ket{\uparrow\downarrow}_7$, that is, for the one-dimensional model, every second site is doubly occupied; for the two-dimensional model, each row has two doubly occupied sites. The entanglement entropy is defined on the subsystem of two sites numbered as $(j, j + 1),~ (j, j + 2), ~(j, j + 3),~ (j, j + 4)$, respectively. In the long-time scale, simulation error in both 1D and 2D cases keeps oscillating. In the 1D case where the entanglement gets close to the maximum $\log4^2=4$, the error oscillates around the Frobenius norm estimate. Whereas in the 2D case, where the entanglement does not increase to the maximum, a gap between the empirical error and the Frobenius norm can be observed.}
    \label{fig:longtimeFH}
\end{figure*}

In fermionic systems, the concept of entanglement fundamentally differs from that in systems of distinguishable spins. Since electrons are identical fermions obeying antisymmetric exchange statistics, the reduced density matrix for a subset of sites cannot be straightforwardly defined by tracing over individual particles. Instead, the entanglement is characterized by the correlations among multiple sites, each hosting a spin-$\frac{1}{2}$ fermion.

Consider a lattice $\Lambda$, with each site $i$ hosting a fermionic mode with spin degrees of freedom described by annihilation operators $c_{i,\sigma}$, where $\sigma = \uparrow, \downarrow$. The global state $\rho$ is represented jointly in the fermionic Fock space and the local internal Hilbert space at each site, respecting the Pauli exclusion principle. To analyze entanglement, one partitions the lattice into a subsystem $A \subseteq \Lambda$ and its complement $B = \Lambda \setminus A$. The reduced density matrix $\rho_A$ is obtained via a fermionic partial trace
$\rho_A = \mathrm{Tr}_B (\rho)$.

Since electrons carry spin, the reduced density matrix for the region $A$ encapsulates correlations over both spatial sites and spin degrees of freedom. The maximal entanglement entropy of a local site ($|A|=1$) should be
$S_A = - \mathrm{Tr} (\rho_A \log \rho_A)$=2. 
This emphasizes that, unlike traditional spin models where the local Hilbert space is simply two-dimensional, in fermionic systems with $1/2$ spin, the local Hilbert space of each site is 4-dimensional, behaving like 2 qubits.

Based on this fact, the analysis in the main text regarding entanglement-induced resilience is equally applicable to fermionic systems. In this context, entanglement is directly defined by the partitioning of sites; increasing entanglement among different sites enhances the overall robustness of the system’s dynamics. Importantly, our focus here is on the dynamical evolution rather than on the ground state problem of the Hamiltonian, which distinguishes this work from \cite{HaiqingLin2004PhysRevLett.93.086402}.

We demonstrate additional numerical results of Fermi-Hubbard model under the perturbative noise (see Eq. \eqref{FH}) with a longer time scale in Figure~\ref{fig:longtimeFH}. Because of the Coulomb interaction and the Pauli exclusion principle, the Fermi–Hubbard model does not fully thermalize; as a result, the simulation error continues to oscillate. This effect is more pronounced in the 2D Fermi-Hubbard model, where the analog simulation error is obviously larger than the Frobenius norm estimate, as shown in Figure~\ref{fig:longtimeFH}.


\subsection{Quantum control}

Consider implementing a $k$-qubit quantum gate, where the $k$ target qubits are coupled to neighboring spectator or intruder qubits. In the physical basis, the total Hamiltonian can be expressed as   $\tilde{H}^{m,k}=\tilde{H}_0+\tilde{V}(t)$, where $H_0$ is the native Hamiltonian and $\tilde{V}(t)$ is the control field.
To analyze coupling effects, we transform to the eigenbasis of of the eigenbasis of $\tilde{H}_0$, where the Hamiltonian can be decomposed into a direct sum of blocks \cite{hai2025scalablerobustquantumcontrol}
\begin{equation}
H^{m;k} = \bigoplus_{i=1}^{2^m} H_i + \bigoplus_{i,j=1}^{2^m} V_{ij}(t)
           \label{Eq_GeneralH}.
\end{equation}
The full matrix of the Hamiltonian is defined in the Supplementary Section B.
Here, the superscript $m;d$ denotes a system with $m$ neighboring qubits affecting $k$ target qubits. The Hamiltonian decomposes into $2^m\times2^m$ blocks corresponding to different configurations of the neighboring qubits. For example, in a system with two neighboring qubits and one target qubit, $\text{span}\{\left|\uparrow\uparrow;\uparrow\right\rangle, \left|\uparrow\uparrow;\downarrow\right\rangle\}$ represents the $i=1$ diagonal subspace where both neighboring qubits are in the $\left|\uparrow\right\rangle$ state.


Without loss of generality, we designate the $i=1$ subspace as our reference and define the reference Hamiltonian $\tilde{H}_{\text{ref}}=\tilde{H}_1+\tilde{V}_{11}$ (here the reference Hailtianonian can be set to the ideal Hamiltonian of implementation of the target quantum gate, i.e., $H_{\text{ref}}=H_{\text{ideal}}$ and $\tilde{H}_{\text{ref}}=\tilde{H}_{\text{ideal}}$). The difference between the reference and the original Hamiltonian is identified as a correlated error Hamiltonian involved with quantum control, which can be seen as one part of the perturbation Hamiltonian:
\begin{equation}
H_{\text{error}} = H^{m;k}-H_{\text{ref}}=\bigoplus_{i=2}^{2^m} (\Delta H_i + \Delta V_{ii}) + \bigoplus_{i,j=1,i\neq j}^{2^m} V_{ij} .
\end{equation}
Here, $\Delta H_i=H_i-H_1$ represents frequency shifts experienced by target qubits due to different configurations of neighboring qubits, collectively forming a manifold of $2^m$ distinct energy levels. The terms $\Delta {V}_{ii}={V}_{ii}-{V}_{11}$ represent variations in control field effects across different neighboring-qubit configurations, while $V_{ij}(t)$ for $i\neq j$ captures crosstalk between different subspaces. 

Let $A$ denote the subsystem consisting of the $m + k$ qubits (target and spectators), and $B$ the remaining $n - m - k$ qubits, with the total state denoted as $\ket{\psi(\tau)}_{AB}$. The initial state before implementing the quantum gate is $\ket{\psi(0)}_{AB}$.
Since the error Hamiltonian $H_{\text{error}}(\tau)$ in quantum control is generally time-dependent and acts locally on subsystem $A$, we now bound the resulting error $\norm{\bigl(U_0(t) - U(t)\bigr) \ket{\psi(0)}_{AB}}$ as follows (see Supplementary Section~B for details):

\begin{equation}
\begin{split}
       & \int_0^t d\tau \norm{H_{\text{error}}(\tau)}_F \\
   & + \int_0^t d\tau \sqrt{\norm{ H_{\text{error}}^\dagger (\tau) H_{\text{error}}(\tau)}\sqrt{2(\log d_A-S(\rho_{A}(\tau)))}}.
\end{split}
\end{equation}
When the entanglement in $\ket{\psi(\tau)}_{AB}$ is sufficiently large during the control period, the error is dominated by the normalized Frobenius norm of the error Hamiltonian, yielding the approximation $\norm{(U_0(t)-U(t))\ket{\psi(0)}}\approx \int_0^t d\tau \norm{H_{\text{error}}(\tau)}_F$.

It is noteworthy that entanglement within subsystem~$A$ (between targets and spectators) also aids in noise suppression. For a clear presentation, here we treat all components inside~$A$ as a unified subsystem, instead of separating them into multiple subspaces corresponding to different $H_{\text{error}}^{(j)}(\tau)$.\\

\textbf{Acknowledgment:}
  Q.Z. acknowledges funding from Innovation Program for Quantum Science and Technology via Project 2024ZD0301900, National Natural Science Foundation of China (NSFC) via Project No. 12347104 and No. 12305030, Guangdong Basic and Applied Basic Research Foundation via Project 2023A1515012185, Hong Kong Research Grant Council (RGC) via No. 27300823, N\_HKU718/23, and R6010-23, Guangdong Provincial Quantum Science Strategic Initiative No. GDZX2303007, HKU Seed Fund for Basic Research for New Staff via Project 2201100596.  X. L. thanks the National Basic Research Program of China (Grants No. 2021YFA1400900). X.-H.D. thanks the Shenzhen Science and Technology Program (KQTD20200820113010023) and the Innovation Program for Quantum Science and Technology (2024ZD0300400).
\textbf{Code availability:}           
The code used in this study is available via GitHub: https://github.com/minidas/analog.


\bibliography{ref,ref_2}

\onecolumngrid
\appendix
\newpage
\newcommand{\appendixtoc}{
  \begingroup
  \let\cleardoublepage\clearpage
  \let\clearpage\relax
  \let\appendixsection\section
  \renewcommand{\section}{\addtocontents{toc}{\protect\setcounter{tocdepth}{-1}}\appendixsection}
  \tableofcontents
  \endgroup
  \addtocontents{toc}{\protect\setcounter{tocdepth}{2}}
}



\section*{Supplementary Materials of "Entanglement-Induced Resilience of Quantum Dynamics}

\section{General quantum dynamics of Time-dependent Hamiltonian with entanglement} \label{timedependent}

By utilizing typical information about the states, one can achieve a more accurate quantification of the erroneous behaviors in quantum dynamics. 
We demonstrate that entanglement inherently provides a level of robustness against dynamic errors, allowing us to derive tighter and more reliable error analyses.
This insight not only underscores the significance of entangled states in enhancing the fidelity of quantum simulations but also highlights the broader applicability. 

Before exploring the role of entanglement in quantum dynamics, we introduce the following mathematical lemma:
\begin{lemma}     \label{HHKL}\cite{HHKL8555119}
    Let $A(t)$ and $B(t)$ be continuous time-dependent Hermitian operators, and let $U_A(t)$ and $U_B(t)$  with $U_A(0) = U_B(0) = I$ be the corresponding time evolution unitaries. 
    Then $W(t)=U_B(t)^\dagger U_A(t)$ is the unique solution of $ i \partial t W(t) =U_B (t)^\dagger (A(t)-B(t))U_B(t ) W(t)$, and $W(0)=I$.
\end{lemma}

Now suppose $A(t)=H_0(t)$ is the ideal Hamiltonian and $B(t)=H^\prime (t)$ which is the noisy Hamiltonian, we have $W(t)=U_{B}(t)^\dagger U_A(t)=\mathcal{T}e^{\int_{s_1=t}^{0}ds_1iH^\prime(s_1)} \mathcal{T}e^{-\int_{s_2=0}^{t}d{s_2}iH_0(s_2)}$. Note that we would like to quantify the distance between $U_A(t)$ and $U_B(t)$ with initial state $\ket{\psi}$. By making use of Lemma \autoref{HHKL}, we now have

\begin{equation}
\begin{split}
        \norm{(U_A(t)-U_B(t))\ket{\psi}}&=\norm{U_B(t)( W(t)-I)\ket{\psi}}\\
       &=\norm{( W(t)-I)\ket{\psi}}\\
       &=\norm{( W(t)- W(0))\ket{\psi}}\\
        &=\norm{i\int^t_0 dt  \partial \tau W(\tau)\ket{\psi}}\\
        &= \norm{\big[\int_0 ^t d\tau U_B (\tau)^\dagger (A(\tau)-B(\tau))U_B(\tau) W(\tau) \big]\ket{\psi}}\\
           &= \norm{\big[\int_0 ^t d\tau U_B (\tau)^\dagger (A(\tau)-B(\tau))U_A \big]\ket{\psi}}\\
        &\le \int_0 ^t d\tau\norm{\big[  (A(\tau)-B(\tau))U_A(\tau)\big]\ket{\psi}}\\
        &=\int_0 ^t d\tau \norm{\big[  H^\prime (\tau)-H_0(\tau)\big]e^{-\int_s^{\tau} d siH(s)}\ket{\psi}}\\
        &=\int_0 ^t d\tau \norm{\big[  H^\prime (\tau)-H_0(\tau)\big]\ket{\psi(\tau)}}.
        \end{split}
\end{equation}
Here we have use $\norm{U A\ket{\psi }}=\norm{ A\ket{\psi }}$ since $U$ is unitary.
One can see that the error term of the upper bound now becomes $ H^\prime (\tau)-H(\tau)$ in the case of a time-dependent Hamiltonian, where $\ket{\psi(\tau)}$ is the ideal evolved state for evolution time $\tau$. We now summarise the above result as the following Corollary:

\begin{corollary} [Bound on the difference between evolved states with time-dependent Hamiltonians]
Let $\ket{\psi(0)}$ be an initial state, and consider two time-dependent Hamiltonians $H_0(\tau)$ and $H'(\tau)$. The vector norm of the difference of two evolved states satisfies 
    \begin{equation}
        \norm{ (\mathcal{T}e^{-i\int_{0}^{t}d{\tau}H_0(\tau)}   -\mathcal{T}e^{-i\int_{0}^{t}d{\tau}H^\prime(\tau)}) \ket{\psi(0)} } \le \int_0 ^t d\tau \norm{\big[  H^\prime (\tau)-H_0(\tau)\big]\ket{\psi(\tau)}},
    \end{equation}
    where $\mathcal{T}$ denotes time-ordering and  $\ket{\psi(\tau)}$ is the state evolved under the Hamiltonian $H'(\tau)$. \label{Cor:bound}
\end{corollary}

This corollary provides a quantitative measure of how differences in the Hamiltonians influence the evolution of the quantum state over time.


\subsection{Entanglement-based bound of quantum dynamics}

Here, we investigate how entanglement influences the upper bound of error for general quantum dynamics.\\

\emph{Time-independent perturbation.} Suppose $H^\prime(\tau)=H(\tau)+H_{\text{pert}}$ and $H_{\text{pert}}$ is time-independent and $H_{\text{pert}}=\sum_k \delta_k V_k+\sum_m \eta_m N_m=\sum_j {H_{\text{pert}}}_j$ where $V_k$  and $N_m$ are Pauli operators. Note that ${H_{\text{pert}}}$ is Hermitian, thus
we have

\begin{equation}
\label{entbound}
\begin{split}
         \norm{(U_A(t)-U_B(t))\ket{\psi}}&\le \int_0 ^t d\tau \norm{\big[  H^\prime(t) -H_0(t)\big]\ket{\psi(\tau)}}\\
         &=\int_0 ^t d\tau  \norm{H_{\text{pert}}\ket{\psi(\tau)}}\\
                 &=\int_0 ^t d\tau \sqrt{\bra{\psi(\tau)} H_{\text{pert}}^\dagger  H_{\text{pert}} \ket{\psi(\tau)}}\\
                           &\le \int_0 ^t d\tau 
         \sqrt{\norm{H_{\text{pert}}}^2_F+\sum_{j,j^\prime} \norm{{H_{\text{pert}}}^\dagger_j{H_{\text{pert}}}_{j^\prime}}\sqrt{2 \log (d_{\support({H_{\text{pert}}}^\dagger_j{H_{\text{pert}}}_{j^\prime})})-2S(\rho_{j,j^\prime})}}\\
    &\le t\norm{H_{\text{pert}}}_F+\int_0 ^t d\tau  \sqrt{\sum_{j,j^\prime} \norm{{H_{\text{pert}}}^\dagger_j{H_{\text{pert}}}_{j^\prime}}\sqrt{2 \log (d_{\support({H_{\text{pert}}}^\dagger_j{H_{\text{pert}}}_{j^\prime})})-2S(\rho_{j,j^\prime})}},
\end{split}
\end{equation}
where $\rho_{j,j^\prime}:=\Tr_{[N]\setminus \support({H_{\text{pert}}}^\dagger_j{H_{\text{pert}}}_{j^\prime})}(\ket{\psi}\bra{\psi})$ is the reduced density matrix of $\ket{\psi}\bra{\psi}$ on the subsystem of $\support({H_{\text{pert}}}^\dagger_j{H_{\text{pert}}}^\dagger_{j^\prime})$, and $\text{S}(\rho_{j,j^\prime})$ is the entanglement entropy of $\rho_{j,j^\prime}$. Here $\|\cdot\|_{\text{F}}$ denotes the normalized Frobenius norm, and in the penultimate inequality, we have utilized the Lemma \autoref{lemma1} while the final inequality we make use of $\sqrt{a+b}\le \sqrt{a}+\sqrt{b}$ for positive semi-definite $a$ and $b$.

Generally, states tend to be thermalized.
Suppose $t=c$, state entanglement is growing to approximate the average case, and $\tau\ge c$, $\norm{H_{\text{pert}}\ket{\psi(\tau)}}\approx \norm{H_{\text{pert}}}_F$. In this case,
we have

\begin{equation}
\begin{split}
         \norm{(U_A(t)-U_B(t))\ket{\psi}}&\lesssim 
                 \int_0 ^c d\tau \sqrt{\bra{\psi(\tau)} H_{\text{pert}}^\dagger  H_{\text{pert}} \ket{\psi(\tau)}}+(t-c)  \norm{H_{\text{pert}}}_F.\\
\end{split}
\end{equation}

For the average case (Haar random or 1 design), we have

\begin{equation}
    \begin{split}
           \int_{\psi\in \text{Haar}} d\psi \norm{ (U_0(t)-U(t))\ket{\psi}} &\le  \int_{\psi\in \text{Haar}} d\psi\int_0 ^t d\tau \sqrt{\bra{\psi(\tau)} H_{\text{pert}}^\dagger  H_{\text{pert}} \ket{\psi(\tau)}}\\
          &\le  \int_0 ^t d\tau \sqrt{\int_{\psi\in \text{Haar}} d\psi\bra{\psi(\tau)} H_{\text{pert}}^\dagger  H_{\text{pert}} \ket{\psi(\tau)}}\\
          &= \int_0 ^t d\tau \sqrt{\Tr(\frac{H_{\text{pert}}^\dagger H_{\text{pert}}}{d})}=  t \norm{H_{\text{pert}}}_F.
    \end{split}
\end{equation}
 Thus, our entanglement-induced bound matches the average case when $t\ge c$.

Remarkably, there is a special case where the worst case is equal to the average case, i.e., the spectral norm $\norm{H_{\text{pert}}}$ is equal to the normalized Frobenius norm $\norm{H_{\text{pert}}}t$. Specifically, suppose $H_{\text{pert}}^\dagger H_{\text{pert}}= \beta I$ and it is easy to prove that $\beta=\norm{H_{\text{pert}}}^2_F=\norm{H_{\text{pert}}}^2$. In this case, the upper bound of state error becomes

\begin{equation}
\begin{split}
         \norm{(U_A(t)-U_B(t))\ket{\psi}}&\le \int_0 ^t d\tau \lambda \sqrt{\bra{\psi(\tau)} H_{\text{pert}}^\dagger  H_{\text{pert}} \ket{\psi(\tau)}}\\
         &= \int_0 ^t d\tau \lambda  \sqrt{\beta}\sqrt{\bra{\psi(\tau)}I\ket{\psi(\tau)}}=  t \norm{H_{\text{pert}}}_F .
\end{split}
\end{equation}
In this case, the upper bound of error is independent of any states and achieve the average case performance.
One of the sufficient conditions for $ H_{\text{pert}}^\dagger H_{\text{pert}}=\beta I$, is  each pair of ${H_{\text{pert}}}_i$ and ${H_{\text{pert}}}_j$ ($i\ne j$) satisfies $\{ {H_{\text{pert}}}_i,{H_{\text{pert}}}_j\}=0$.\\ 

\emph{Time-dependent perturbation.} In general, $H_{\text{pert}}=H^\prime (\tau)-H_0(\tau)$ may be time-dependent, e.g, $H_{\text{pert}}=H_{\text{pert}}(\tau)$. Here for investigate how entanglement plays the role in the error bound of the time-dependent cases, we define $\Delta t= t/J$ and rewrite the upper bound as
\begin{equation}
  \norm{(U_A(t)-U_B(t))\ket{\psi}} \le \int_0 ^t d\tau \norm{ H_{\text{pert}}(\tau)\ket{\psi(\tau)}}=
   \sum_{j=0}^{J-1} \int_{j\Delta t} ^{(j+1)\Delta t} d\tau \sqrt{\bra{\psi(\tau)} H_{\text{pert}}^\dagger(\tau)  H_{\text{pert}}(\tau) \ket{\psi(\tau)}}. \label{ap:eqtimed}
\end{equation}
Note that $\norm{ H_{\text{pert}}\ket{\psi(\tau)}}^2\le \norm{H_{\text{pert}}(\tau)}_F^2+\Delta_{H^\dagger_{\text{pert}}(\tau) H_{\text{pert}}(\tau)}(\psi(\tau))$ (see Eq. \eqref{entbound}) where $$\Delta_{H^\dagger_{\text{pert}}(\tau) H_{\text{pert}}(\tau)}(\psi(\tau))=\sum_{j,j^\prime} \norm{{H_{\text{pert}}}^\dagger_j(\tau){H_{\text{pert}}}_{j^\prime}(\tau)}\sqrt{2 \log (d_{\support({H_{\text{pert}}}^\dagger_j(\tau){H_{\text{pert}}}_{j^\prime}(\tau))})-2S(\rho_{j,j^\prime})}).$$
Define $A^{(j+1)\Delta t}_{j\Delta t}=\int_{j\Delta t} ^{(j+1)\Delta t} d\tau H_{\text{pert}}^\dagger(\tau)  H_{\text{pert}}(\tau) $, a positive operator for segment $j$, we bound the error as

\begin{equation}
\begin{split}
   (\Delta t)^{\frac{1}{2}}\sum_{j=0}^{J-1} \sqrt{\frac{\Tr(A^{(j+1)\Delta t}_{j\Delta t})}{d}}+
   \sum_{j=0}^{J-1}\int_{j\Delta t} ^{(j+1)\Delta t} d\tau  \sqrt{\Delta_{{H^\dagger_{\text{pert}}(\tau)H_{\text{pert}}(\tau)}} (\psi(\tau))}. 
   \end{split}
   \label{timed}
\end{equation}
We have make use of the inequality $\sqrt{x^2+y^2}\le x+y $ for positive numbers $x$ and $y$, and the Cauchy-Schwarz inequality $\int_a^bf(x)g(x)dx\le \sqrt{\int_a^bf^2(x)dx\int_a^bg^2(x)dx}$.
Suppose $\ket{\psi}(\tau)$ have enough entanglement all the time, one has 
\begin{equation}
    \norm{(U_A(t)-U_B(t))\ket{\psi}}\lesssim  \sum_{j=0}^{J-1} (\Delta t)^{\frac{1}{2}}\sqrt{\frac{\Tr(A^{(j+1)\Delta t}_{j\Delta t})}{d}}=\int_{0}^t \norm{H_{\text{pert}}(\tau)}_F.
\end{equation}

It is easy to verify that, if $H_{\text{pert}}(\tau)=H_{\text{pert}}$, one has $A^{(j+1)\Delta t}_{j\Delta t}=\Delta t H^\dagger_{\text{pert}}H_{\text{pert}} $,
and the Eq. \eqref{timed} can reduce to 

\begin{equation}
     t\norm{H_{\text{pert}}}_F+\int_{0} ^{t} d\tau  \sqrt{\Delta_{{H^\dagger_{\text{pert}}(\tau)H_{\text{pert}}(\tau)}} (\psi(\tau))},
\end{equation}
which is the exact upper bound of the total error for the time-independent perturbation case.

\subsection{Trace distance with statistical disorder}
The statistical properties of $\delta_k$ may further cancel some errors in quantum dynamics.
Recall that $ H_{\text{pert}} = \sum_k \delta_k V_k + \sum_m \eta_m N_m$ and $\delta_k$ is the random real parameter of  Gaussian distribution. Assuming that the  Gaussian distribution has $\mathbb{E} \delta _k=0$ and $\text{Var}(\delta_k)=\sigma^2_k$.  Note that the upper bounds of the above analysis are for a single run with fixed $\delta_k$. By considering the unbiased perturbative nature of $\delta_k$, the density matrix of the real evolved state is given as 
$$\rho(t)= \mathbb{E}_{\{\delta_1,\delta_2,...\}}U (t)\ket{\psi}\bra{\psi}U^\dagger(t),$$ where 
$U(t)=\mathcal{T}e^{-i\int H^\prime(t) dt}$ and $H^\prime=H_0+H_{\text{pert}}$.
Now, the distance between the ideal evolution and the approximated one should be defined by the 1-norm $\norm{\rho_0(t)-\rho(t)}_1$ ($\norm{A}_1=\Tr(\sqrt{A^\dagger A})$), not the vector norm. Here $\rho_0(t)=U_0(t)\ket{\psi}\bra{\psi}U^\dagger_0(t) $.

The trace distance between the ideal evolution state and the approximate evolution state is 
\begin{equation}
\begin{split}
\norm{ \rho_0(t)-\rho(t) }_1 &=\|U_0(t)\ket{\psi}\bra{\psi}U^\dagger _0(t) -  \mathbb{E}_{\{\delta_1,\delta_2,...\}} U (t)\ket{\psi}\bra{\psi}U^\dagger(t)\|_1\\
&\le   \mathbb{E}_{\{\delta_1,\delta_2,...\}} \|U_0 (t)\ket{\psi}\bra{\psi}U^\dagger_0(t) - U (t)\ket{\psi}\bra{\psi}U^\dagger(t)\|_1\\
& \le 2   \mathbb{E}_{\{\delta_1,\delta_2,...\}} \norm{(U_0(t)-U(t))\ket{\psi}} \\
& \le 2  \mathbb{E}_{\{\delta_1,\delta_2,...\}}\int_0 ^t d\tau \sqrt{\bra{\psi(\tau)} H_{\text{pert}}^\dagger H_{\text{pert}} \ket{\psi(\tau)}}\\
& \le 2 \int_0 ^t d\tau \sqrt{\bra{\psi(\tau)}   \mathbb{E}_{\{\delta_1,\delta_2,...\}}H_{\text{pert}}^\dagger H_{\text{pert}} \ket{\psi(\tau)}}.
\end{split}
\end{equation}
 For the first inequality, we use the triangle inequality. For the second inequality, for any two projectors $\ket{\chi}\bra{\chi}$ and $\ket{\phi}\bra{\phi}$, the following inequality holds, e.g., $\norm{\ket{\chi}\bra{\chi}-\ket{\phi}\bra\phi{}}_1\le 
2 \norm{\ket{\psi}-\ket{\phi}}$. In the last inequality, we utilize the Cauchy-Schwarz inequality. 

Generally, $H_{\text{pert}} = V+N=\sum_k\delta_k V_k+\sum_m\eta_m N_m$, one has

\begin{equation}
\begin{split}
    \mathbb{E}_{\{\delta_1,\delta_2,...\}}H_{\text{pert}}^\dagger H_{\text{pert}}&=  \mathbb{E}_{\{\delta_1,\delta_2,...\}} (\sum_k \delta_k V_k + \sum_m \eta_m N_m)^\dagger(\sum_{k^\prime} \delta_{k^\prime} V_{k^\prime} + \sum_{m^\prime} \eta_{m^\prime} N_{m^\prime})    \\
    &= \mathbb{E}_{\{\delta_1,\delta_2,...\}}(\sum_{k,k^\prime}\delta_k\delta_{k^\prime} V^\dagger_kV_{k^\prime}+\sum_{k,m^\prime}\delta_k\eta_{m^\prime} V^\dagger_kN_{m^\prime}+\sum_{m,k^\prime}\eta_m\delta_{k^\prime} N^\dagger_mV_{k^\prime}+\sum_{m,m^\prime}\eta_m\eta_{m^\prime} N^\dagger_mN_{m^\prime})\\
    &=\mathbb{E}_{\{\delta_1,\delta_2,...\}}\sum_{k,k^\prime}\delta_k\delta_{k^\prime} V^\dagger_kV_{k^\prime}+\sum_{m,m^\prime}\eta_m\eta_{m^\prime} N^\dagger_mN_{m^\prime}\\
    &=\mathbb{E}_{\{\delta_1,\delta_2,...\}}\sum_{k}\delta^2_k V^\dagger_kV_{k}+\sum_{m,m^\prime}\eta_m\eta_{m^\prime} N^\dagger_mN_{m^\prime}\\
    &=\mathbb{E}_{\{\delta_1,\delta_2,...\}}\sum_{k}\delta^2_k V^\dagger_kV_{k}+N^\dagger N.\\
\end{split}
\end{equation}
Suppose $\{V_k\}$ are Pauli operators (up to a global phase), 
\begin{equation}
     \mathbb{E}_{\{\delta_1,\delta_2,...\}}H_{\text{pert}}^\dagger H_{\text{pert}}=\mathbb{E}_{\{\delta_1,\delta_2,...\}}\sum_{k}\delta^2_k I+N^\dagger N.
\end{equation}

Now the distance of the quantum state is given as
\begin{equation}
\begin{split}
\norm{ \rho_0(t)-\rho(t) }_1 &\le 2 \int_0 ^t d\tau \sqrt{\bra{\psi(\tau)} \mathbb{E}_{\{\delta_1,\delta_2,...\}}( H_{\text{pert}}^\dagger H_{\text{pert}} )\ket{\psi(\tau)}}\\
&=2 \int_0 ^t d\tau \sqrt{\bra{\psi(\tau)} (\mathbb{E}_{\{\delta_1,\delta_2,...\}}\sum_{k}\delta^2_k I+N^\dagger N)\ket{\psi(\tau)}}\\
&=2 \int_0 ^t d\tau \sqrt{\mathbb{E}_{\{\delta_1,\delta_2,...\}}\sum_{k}\delta^2_k+\bra{\psi(\tau)}  N^\dagger N\ket{\psi(\tau)}}\\
&= 2 \int_0 ^t d\tau \sqrt{\sum_k \sigma^2_k+\bra{\psi(\tau)}  N^\dagger N\ket{\psi(\tau)}},\\
&\le 2t \sqrt{\sum_k \sigma^2_k} +2\int_0 ^t d\tau \sqrt{\bra{\psi(\tau)}  N^\dagger N\ket{\psi(\tau)}},\label{eq:disorder},
\end{split}
\end{equation}
where $\sigma_k=\delta^2_k$ is the variance of the Gaussian distribution. Eventually, one has the entanglement-based bound for the density matrix case, i.e.,

\begin{equation}
    \norm{ \rho_0(t)-\rho(t) }_1\le 2t \sqrt{\sum_k \sigma^2_k}+ 2t\norm{N}_F+2 \int_0^t \tau \sqrt{\Delta_{N^\dagger N}(\psi(\tau))},
\end{equation}
where $\Delta_{N^\dagger N}(\psi(\tau))=\sum_{j,j^\prime} \norm{{N}^\dagger_j{N}_{j^\prime}}\sqrt{2 \log (d_{\support({N}^\dagger_j{N}_{j^\prime})})-2S(\rho_{j,j^\prime}(\tau))}$.

\section{The entanglement role in quantum control}


In this section, we specifically focus on the quantum control with an entangled input state, which is ubiquitous in the realm of quantum many-body physics. We demonstrate that entanglement inherently provides a level of robustness against quantum control errors, allowing us to derive tighter and more reliable error analyses.

Suppose one wants to implement a $k$-qubit quantum gate. In general, the $k$ target qubit is coupled to neighboring qubits (spectators or intruders). The total Hamiltonian with physical basis can be expressed as   $\tilde{H}=\tilde{H}_0+\tilde{V}(t)$ where $\tilde{H}_0$ is the native Hamiltonian and $\tilde{V}(t)$ is the control field.
To analyze coupling effects in the eigenbasis  of $\tilde{H}_0$, the Hamiltonian can be decomposed into blocks \cite{hai2025scalablerobustquantumcontrol}
\begin{equation}
H^{m;k} = \bigoplus_{i=1}^{2^m} H_i + \bigoplus_{i,j=1}^{2^m} V_{ij}(t)=\begin{pmatrix}
            H_1 + V_{11} & V_{12} & ... & V_{1m} \\
            V^\dagger_{12} & H_2 + V_{22}& ... & V_{2m} \\
            ... & ... & ... & ...\\
            V^\dagger_{1m} & V^\dagger_{2m}  & ... & H_m +V_{mm}
           \end{pmatrix}.   
           \label{Eq_GeneralH}
\end{equation}
Here, the superscript $m;d$ denotes a system with $m$ neighboring qubits affecting $k$ target qubits. The Hamiltonian decomposes into $2^m\times2^m$ blocks corresponding to different configurations of the neighboring qubits . For example, in a system with two neighboring qubits and one target qubit, $\text{span}\{\left|\uparrow\uparrow;\uparrow\right\rangle, \left|\uparrow\uparrow;\downarrow\right\rangle\}$ represents the $i=1$ diagonal subspace where both neighboring qubits are in the $\left|\uparrow\right\rangle$ state.

\begin{figure}
    \centering
    \includegraphics[width=0.8\linewidth]{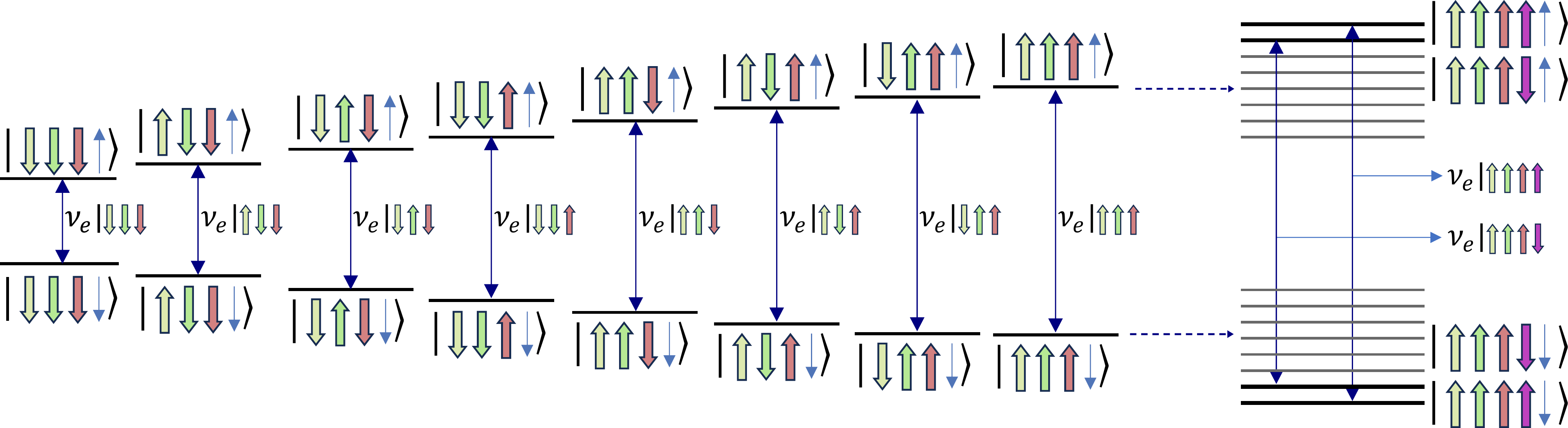}
    \caption{Spectator-induced energy-level splitting. The interaction between the target-controlled qubit and spectators causes the energy levels to split into energy bands. As a result, the spectators lead to unavoidable errors in the implementation of a quantum gate.}
    \label{fig:energy_band}
\end{figure}

Without loss of generality, we designate the $i=1$ subspace as our reference and define the reference Hamiltonian $H_{\text{ref}}=H_1+V_{11}$. The difference between the reference and the original Hamiltonian is identified as the noise Hamiltonian:
\begin{equation}
H_{\text{noise}} = H^{m;k}-H_{\text{ref}}=\bigoplus_{i=2}^{2^m} (\Delta H_i + \Delta V_{ii}) + \bigoplus_{i,j=1,i\neq j}^{2^m} V_{ij}=\begin{pmatrix}
            0 & V_{12} & ... & V_{1m} \\
            V^\dagger_{12} & \Delta H_2 + \Delta V_{22}& ... & V_{2m} \\
            ... & ... & ... & ...\\
            V^\dagger_{1m} & V^\dagger_{2m}  & ... & \Delta H_m +\Delta V_{mm}
           \end{pmatrix}   
\end{equation}
Here, $\Delta H_i=H_i-H_1$ represents frequency shifts experienced by target qubits due to different configurations of neighboring qubits, collectively forming a manifold of $2^m$ distinct energy levels (as shown in Figure \ref{fig:energy_band}). The terms $\Delta V_{ii}=V_{ii}-V_{11}$ represent variations in control field effects across different neighboring-qubit configurations, while $V_{ij}(t)$ for $i\neq j$ captures crosstalk between different subspaces. These noise terms introduce correlated errors, parasitic operations, and cross-coupling effects that degrade gate fidelities.

\subsection{Error analysis of quantum control with a single two-party partition}

Recall that $H_{\text{error}}$ is time-dependent in quantum control protocols, e.g, $H_{\text{error}}=H_{\text{error}}(\tau)$. 
Here, for simplicity to show the important role of entanglement in quantum control, we treat $H_{\text{error}}^\dagger (
       \tau ) H_{\text{error}}(\tau)$ as a single term. Denote $A$ as the subsystem of target qubits and spectators, and $B$ as the remaining qubits. In this setting, $H_{\text{error}}^\dagger (
       \tau ) H_{\text{error}}(\tau)$ can be consider a local error term in the subsystems $A$. 
Now, by using Lemma \ref{lemma1} (single term case), the error upper bound of the implementation of a quantum gate with state $\ket{\psi(0)}=\ket{\psi(0)}_{AB}$ is given as 

\begin{equation}
\begin{split}
       \norm{(U_0(t)-U(t))\ket{\psi(0)}_{AB}} &\le  \int_0^t d\tau  \sqrt{\bra{\psi(\tau)}_{AB}H_{\text{error}}^\dagger (
       \tau
       ) H_{\text{error}}(\tau)\ket{\psi(\tau)}_{AB}}\\
       &\le \int_0^t d\tau \sqrt{\norm{H_{\text{error}}}^2_F+ \norm{ H_{\text{error}}^\dagger  H_{\text{error}}}\sqrt{2(\log d_A-S(\rho_{A}(\tau)))}} \\
       & \le \int_0^t d\tau \norm{H_{\text{error}}(\tau)}_F + \int_0^t\sqrt{\norm{ H_{\text{error}}^\dagger (
       \tau
       ) H_{\text{error}}(\tau)}\sqrt{2(\log d_A-S(\rho_{A}(\tau)))}}.
\end{split}
\end{equation}
 In the main text, we trace the upper bound of quantum gate error with quantum control by $\int_0 ^t d\tau \sqrt{\bra{\psi(\tau)} H_{\text{error}}^\dagger  H_{\text{error}} \ket{\psi(\tau)}}$. 
 For the modest system sizes considered in quantum control protocols, this quantity can be calculated efficiently. 
 Clearly, if the entanglement of $\ket{\psi (\tau)}_{AB}$ is large enough at the period of quantum control, one has $\norm{(U_0(t)-U(t))\ket{\psi}}\approx \int_0^t d\tau \norm{H_{\text{error}}(\tau)}_F$.
 
It is important to note that entanglement can also occur between the target qubits and the nearby spectators within subsystem A, which can contribute to noise suppression. To clearly and simply demonstrate the role of entanglement in quantum control, we categorize these interactions into a single class, e.g., treating the target and spectators as a subsystem, rather than subdividing them into many subspaces composed of error terms $H_{\text{error}}(\tau)_j$.




\subsection{Error analysis with interaction picture}

To more intuitive to show how entanglement suppress the evolved state error after evolution, one can decompose the total evolution operator as $U = U_0 U_{\text{e}}$, separating the ideal evolution $U_0(t) = \mathcal{T} \exp\{-i \int_0^t d\tau H_{\text{ref}}(\tau)\}$ from the error evolution $U_{\text{e}}(t) = \mathcal{T} \exp\{-i \int_0^t d\tau V_{\text{I}}(\tau)\}$. Here, $V_{\text{I}}= U_0^{\dagger} H_{\text{error}} U_0$ represents the noise Hamiltonian in the interaction picture. 
For any operator $K$, we define the super-operator $\mathcal{R}(K(t))=\int_0^t d\tau U_0(\tau)^{\dagger} K(\tau) U_0(\tau)$, which represents the time-integrated effect of $K$ as transformed by the ideal evolution. In the Pauli basis $\{\sigma^k_\nu\}=\{X,Y,Z,I\}^{\otimes k}$, an operator $K=\mathbf{T}\cdot \mathbf{\sigma}^k$ can be conceptualized as a point moving with velocity $\mathbf{T}$ in operator space, and $\mathcal{R}(K(t))=\mathbf{r}\cdot \mathbf{\sigma}^k$ represents the integrated path traced by this point in the Pauli frame~\cite{hai2025geometric}.

Under the assumption that the noise terms are small in comparison to the ideal Hamiltonian, one can employ the first-order Magnus expansion to eliminate time-ordering, resulting in a product of error evolution operators corresponding to each error term. Additionally, we applied the first order of the Taylor series to derive:
\begin{equation}
\begin{aligned}
U_{\text{e}}(t) &=\mathcal{T} \exp\{-i \mathcal{R} H_{\text{error}}(t))\}\approx \prod_{j,k=1}^{2^m}e^{-i\mathcal{R}(\Delta H_j)}e^{-i\mathcal{R}(\Delta V_j)}e^{-i\mathcal{R}(V_{jk})}\\
&\approx I - i\sum_{\mu} \mathcal{R}(K_\mu)=I - i\sum_{\mu}\mathbf{r}_\mu(t)\cdot \mathbf{\sigma}^{m+k},
\end{aligned}
\label{Eq_GeneralUe}
\end{equation}
where each noise term $K_\mu$ belongs to the set $\{\Delta H_j,\Delta V_j, V_{jl} \}$. Each error component generates an error curve $\mathbf{r}_\mu(t)$ in operator space, with the vector norm $\|\mathbf{r}_\mu(t)\|$ quantifying the magnitude of the associated error. Note that here $\sigma^k$ denotes the Pauli operator for the target $k$ qubits. 

Recall that $U_e(t)=I+\mathscr{M}(t)\approx I_{md} - i\sum_{\mu}\mathbf{r}_\mu(t)\cdot \mathbf{\sigma}^{m+k}$, now the additive error of the ideal unitary $U_0(t)$ and the approximate unitary $U(t)=U_0(t)U_e(t)$ for given initial state $\ket{\psi}$ is given as

\begin{equation}
\begin{split}
       \norm{ (U_0(t)-U(t))\ket{\psi}}&= \norm{U_0(t) (I-U_e)\ket{\psi}}\\
        &\approx
        \norm{U_0(t)(i\sum_{\mu}\mathbf{r}_\mu(t)\cdot \mathbf{\sigma}^{m+k}) \ket{\psi}}=\norm{(\sum_{\mu}\mathbf{r}_\mu(t)\cdot \mathbf{\sigma}^{m+k}) \ket{\psi}}.
\end{split}
\end{equation}
The average error is bounded by

\begin{equation}
\begin{split}
      \int_{\psi\in \text{Haar}} d\psi \norm{ (U_0(t)-U(t))\ket{\psi}}&\approx   \int_{\psi\in \text{Haar}} d\psi \norm{(\sum_{\mu}\mathbf{r}_\mu(t)\cdot \mathbf{\sigma}^{k+m}) \ket{\psi}}\\
      &=\int_{\psi\in \text{Haar}} d\psi \sqrt{\bra{\psi}(\sum_{\mu^\prime}\mathbf{r}_{\mu^\prime}(t)\cdot \mathbf{\sigma}^{k+m}) ^\dagger (\sum_{\mu}\mathbf{r}_\mu(t)\cdot \mathbf{\sigma}^{k+m}) \ket{\psi}}\\
      &\le\sqrt{ \int_{\psi\in \text{Haar}} d\psi \bra{\psi}(\sum_{\mu^\prime}\mathbf{r}_{\mu^\prime}(t)\cdot \mathbf{\sigma}^{k+m}) ^\dagger (\sum_{\mu}\mathbf{r}_\mu(t)\cdot \mathbf{\sigma}^{k+m}) \ket{\psi}}\\
      &=\sqrt{\Tr[(\sum_{\mu^\prime}\mathbf{r}_{\mu^\prime}(t)\cdot \mathbf{\sigma}^{m+k}) ^\dagger (\sum_{\mu}\mathbf{r}_\mu(t)\cdot \mathbf{\sigma}^{m+k})]/d}\\
      &= \norm{\sum_{\mu}\mathbf{r}_\mu(t)\cdot \mathbf{\sigma}^{m+k}}_F.
\end{split}
\end{equation}
 In the above inequality, we use the 1-design properties of Haar-random. This suggests our analysis holds for 1-design ensembles.

In quantum control theory, one usually defines the total error distance as a measure of control robustness against all noises
\begin{equation}
D = 2^{-\frac{m+k}{2}} \sum_\mu \left\|\mathbf{r}_\mu\cdot \mathbf{\sigma}^{m+k}\right\|_{\text{2}}
= \sum_\mu \|\mathbf{r}_\mu(t)\|.
\end{equation}
Here $\|\cdot\|_{\text{2}}$ denotes the 2 norm or unnormalized Frobenius norm. Since $\norm{\sum_{\mu}\mathbf{r}_\mu(t)\cdot \mathbf{\sigma}^{m+k}}_F \le \sum_\mu \norm{\mathbf{r}_\mu(t)}$, this measure is the upper bound of the average performance of multiplicate error for random input states. It also suggests that the transitional bound of quantum control is not tight in the average case (note that this distance is not the worst case either).

For perfect gate implementation with high precision and robustness, we require $U_0(T) = U_{\text{target}}$ or the noiseless gate fidelity and error distance $D=0$. 
Therefore, the error upper bound serves as the error-correcting constraint when designing robust control pulses~\cite{xue2025quantum}.
Here, we show that, in the case of an input state that has sufficient entanglement between the target qubits and the remaining qubits, the error of implementing the target unitary can be reduced to the random case. Since 
$       \norm{ (U_0(t)-U(t))\ket{\psi}} \approx\norm{(\sum_{\mu}\mathbf{r}_\mu(t)\cdot \mathbf{\sigma}^{m+k}) \ket{\psi}}$, by utilizing our entanglement-based lemma,
we have

\begin{equation}
    \norm{ (U_0(t)-U(t))\ket{\psi}} \approx \norm{\sum_{\mu}\mathbf{r}_\mu(t)\cdot \mathbf{\sigma}^{m+k}}_F.
\end{equation}
This suggests that entanglement can induce robustness in implementing quantum gates.

\subsection{Quantum control of coupled quantum dots}

To illustrate our approach, we consider a common example: a pair of coupled gate-defined quantum dots. The spin states follow an extended Heisenberg Hamiltonian~\cite{burkard2023semiconductor}
\begin{equation}
 \tilde{H}= \sum_i\mathbf{B}_i \cdot \mathbf{S}_i + \sum_{<i,j>}J
  \left( \mathbf{S}_i \cdot \mathbf{S}_j - \frac{1}{4} \right).
  \label{Eq_Heisenberg_Hamiltonian}
\end{equation}
Here, $\mathbf{S}_j = (S_{Xj}, S_{Yj}, S_{Zj}) / 2$ are the spin operators, and $\mathbf{B}_j = (B_{x, j}, B_{y, j}, B_{z, j})$ represents the magnetic field at each qubit. The z-components of the magnetic fields determine the electron spin resonance frequencies, while transverse fields provide qubit control. The exchange coupling $J$ brings the two spins into the coupled basis $\{\left|\uparrow \uparrow \right\rangle, |\tilde{\uparrow \downarrow}\rangle, |\tilde{\downarrow \uparrow}\rangle, \left|\downarrow \downarrow \right\rangle\}$, in which the Hamiltonian without transverse controls is diagonalized as $\tilde{H}_0 = \text{diag} \{ 2 E_z, -
\Delta \tilde{E}_z - J, \Delta \tilde{E}_z - J, - 2 E_z \}$, where $\Delta \tilde{E}_z = \sqrt{J^2 + \Delta E_z^2}$ , $|\tilde{\uparrow \downarrow}\rangle=\cos\theta|\uparrow \downarrow\rangle+\sin\theta| \downarrow\uparrow\rangle$, $ |\tilde{\downarrow \uparrow}\rangle=\sin\theta|\uparrow \downarrow\rangle+\cos\theta| \downarrow\uparrow\rangle$, and $\tan \theta = \frac{J}{\Delta E_z + \Delta \tilde{E}_z}$.

When driving qubit 2 with a transverse field and transforming to the rotating frame, we obtain 
\begin{equation}
    H^{1;1} = \begin{pmatrix}
        \tilde{H}_1 + \tilde{V}_{11} & \tilde{V}_{12} \\
        \tilde{V}_{21} & \tilde{H}_2 + \tilde{V}_{22}
    \end{pmatrix},
    \label{Eq_2dot_H2_full}
\end{equation}
where the components are given by $H_{1,2}= \text{diag} \{\pm J/4, \mp J/4 \}$, $\tilde{V}_{11} = \tilde{V}_{22} = \frac{\Omega_{2}}{2} X$, $\tilde{V}_{12,21} = \frac{1}{2} e^{\pm i\tilde{E}_z t} \tan \theta \Omega_2 Z$. In the coupled Pauli basis, it is simplified as
\begin{equation}
    \begin{aligned}
        H^{1;1} &= \frac{\Omega_2(t)}{2} IX + \frac{J}{4} ZZ
         + \frac{\tan\theta \Omega_2(t)}{2} \\ 
         & \times
         \Big[\cos(\Delta \tilde{E}_z t) XZ - \sin(\Delta \tilde{E}_z t) YZ \Big],
    \end{aligned}
    \label{Eq_2dot_H_Pauli}.
\end{equation}
This Hamiltonian contains the intended control term ($IX$) along with parasitic terms: always-on $ZZ$ coupling and crosstalk terms ($XZ$, $YZ$). Besides the coupling-induced terms, the total noise Hamiltonian also includes frequency drifts $\delta$ and control amplitude fluctuations $\epsilon$~\cite{yi2024robust}. Hence, the total perturbation Hamiltonian is
\begin{equation}
    \begin{aligned}
       H^{1,1}_{\text{pert}} &= \delta IZ + \delta ZI + \epsilon\Omega_2(t)IX + \frac{J}{4} ZZ
         + \frac{\tan\theta \Omega_2(t)}{2}
         \Big[\cos(\Delta \tilde{E}_z t) XZ - \sin(\Delta \tilde{E}_z t) YZ \Big].
    \end{aligned}
    \label{Eq_2dot_H_Pauli}
\end{equation}

This perturbative term can be generalized to the multi-spectator case. Given  $\{\mathrm{m}, \mathrm{1} \}$ spectators and target qubits, the total perturbation Hamiltonian that induces errors on the target qubit is 
\begin{equation}
    \begin{aligned}
        H_{\text{pert}}^{m,1} &=\sum_i \delta Z_i + \epsilon\Omega_2(t)X_1 + \sum_i\frac{J}{4} Z_i Z_1
         + \sum_i\dfrac{\tan\theta \Omega_2(t)}{2}
         \Big[\cos(\Delta \tilde{E}_z t) X_iZ_1 - \sin(\Delta \tilde{E}_z t) Y_iZ_1 \Big].
    \end{aligned}
    \label{Eq_multidot_H_Pauli}
\end{equation}

In experimentally relevant regimes where $J \ll \Delta E_z$, we consider $J/\Delta E_z < 0.1$ and get $\tan \theta \approx J/2\Delta E_z$. The magnitudes of crosstalk noise and always-on coupling scale as $|\tan\theta \Omega_2|$ and $J$, respectively.
While crosstalk is often negligible in weak-drive regimes~\cite{russ2018high, gungordu2020robust, kanaar2021single}, it becomes significant under strong-drive conditions needed for fast gate operations, thereby necessitating the simultaneous suppression of both error channels~\cite{heinz2024analysis, mkadzik2022precision} for high-fidelity operations in multi-qubit systems.

\subsection{Two-qubit gate with quantum control}
For a two-qubit gate, $\{m,2\}$ indicates two target qubits and $m$ spectators, depending on lattice geometry. For example, in a square lattice, $m=6$. The Hamiltonian to implement a $\sqrt{SWAP}$ gate \cite{he2019two} in the raw basis (computational basis) relies on the standard homogeneous Heisenberg interaction $\tilde{H}= \sum_i\mathbf{B} \cdot \mathbf{S}_i + \sum_{<i,j>}J
  \left( \mathbf{S}_i \cdot \mathbf{S}_j - \frac{1}{4} \right)$ by taking $\mathbf{B}_1=\mathbf{B}_2$. Correspondingly, the entangling gate is $U=diag\{1,-1,1,1\}$ in the coupled basis $\{\left|\uparrow \uparrow \right\rangle, |\tilde{\uparrow \downarrow}\rangle, |\tilde{\downarrow \uparrow}\rangle, \left|\downarrow \downarrow \right\rangle\}$, and the ideal Hamiltonian is
\begin{equation}
    \begin{aligned}
       H^{m,2} &= -\frac{\omega}{2} \sum_{j=1,2} Z_j + \frac{J}{4} Z_1 Z_2 - \sum_{i \in \{m\}}\frac{\omega_i}{2} Z_i.
    \end{aligned}
\end{equation}

Noticing that the two target qubits are on resonance and decoupled from any other spectators, the Hamiltonian in the rotating frame, with all local spins, therefore becomes
\begin{equation}
    \begin{aligned}
       H_{\text{rot}}^{m,2} &= \frac{J}{4} Z_1 Z_2.
    \end{aligned}
    \label{eq:2qubit_rot}
\end{equation}

Eventually, the overall error term of the Hamiltonian for this system, including the spectators, is
\begin{equation}
    \begin{aligned}
       H_{\text{pert}}^{m,2} &=\sum_{j\in{m,2}} \delta_j Z_j + \sum_{i \in \{m\}, j=1,2}\frac{J_{res}}{4} Z_i Z_j.
       \label{eq:2qubit_per}
    \end{aligned}
\end{equation}
where $\delta_j$ are variations in qubit energies, $J_{res}$ are unwanted couplings between targets and spectators, while $J$ is the coupling strength tuned to implement the $\sqrt{SWAP}$ gate.




\section{More Numerical results and details }
Here we demonstrate numerical results to illustrate the role that entanglement plays in both analog simulation of the quantum Ising model with a transverse field (QIMF) and quantum control in the quantum dot model.

\subsection{QIMF models}

\emph{1D-QIMF.---} First, we consider a simple 1-D QIMF with Hamiltonian
\begin{equation}
    H_0=h_x\sum_{i=1}^NX_i+h_y\sum_{i=1}^NY_i+J\sum_iX_iX_{i+1},
\end{equation}
where the parameters $h_x=0.809,h_y=0.9045,J=1$. There is a perturbation term in actual analog simulation $H_{\text{analog}} =H_0+H_{\text{dis}}+ H_{\text{pert}}$. The source of perturbation includes stochastic disorder and
intrinsic control imperfections. In our numerical tests, the disorder term $H_{\text{dis}}$ is set as a combination of $N$ single-qubit Pauli operators $X_i$'s, i.e., $H_{\text{dis}}=\sum_{i=1}^N\delta_iX_i$, where $\delta_i$ is randomly chosen with respect to a normal distribution and $X_i$ represents Pauli $X$ acting on the $i$th qubit. The control disorder term is set as $H_{\text{imp}}=\eta\sum_{i=1}^{N-1}X_iX_{i+1}$, where $\eta$ is a constant. We demonstrate numerical tests of three different scenarios: with disorder only ($\delta_i\in\mathcal N(0,0.001),\eta=0$), with imperfection only ($\delta_i=0,\forall i$ and $\eta=0.001$), and with both disorder and imperfection terms ($\delta_i\in\mathcal N(0,0.001),\eta=0.001$). We compare the simulation error, defined as $$\|e^{-iH_{\text{analog}}t}\ket{\psi(0)}-e^{-iH_0t}\ket{\psi(0)}\|,$$ for a typical initial state versus an atypical initial state in Figure \ref{fig:placeholder}. In typical cases ($\ket{\psi(0)}=\ket0^{\otimes N}$), entanglement entropy increases with the evolution until saturation, whereas in atypical cases ($\ket{\psi(0)}=\ket+^{\otimes N}$), it remains at a low level (See Figure \ref{fig:placeholder}(b)(d)(e)). Consequently, in typical scenarios, the simulation error rapidly decreases and stabilizes near the average estimate, whereas in atypical scenarios, the error remains larger and may approach the worst-case estimate given by the spectral norm. 
\begin{figure}
    \centering
    \subfloat[]{
    \includegraphics[width=0.42\linewidth]{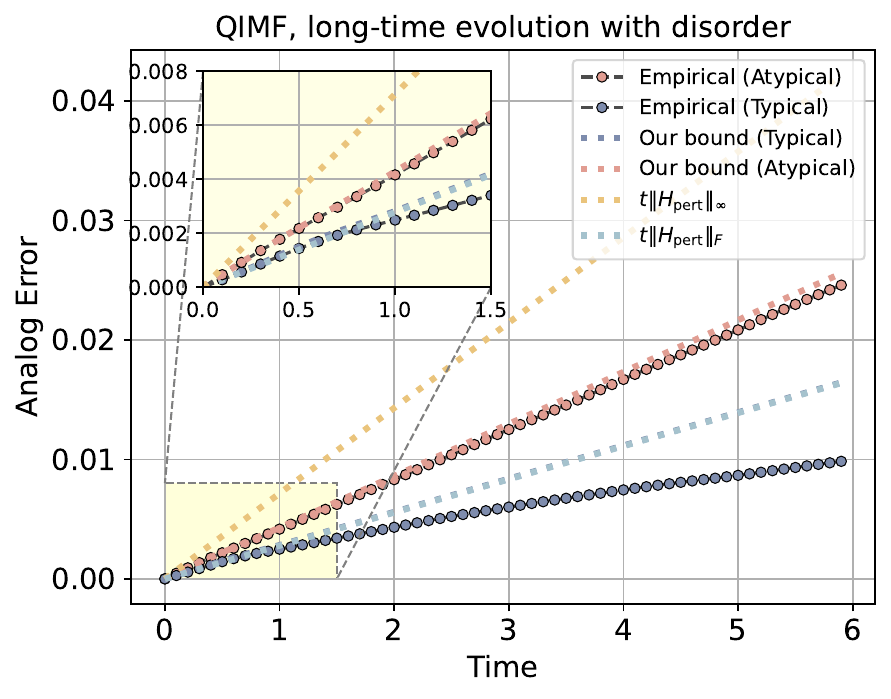}
    }
    \subfloat[]{
    \includegraphics[width=0.5\linewidth]{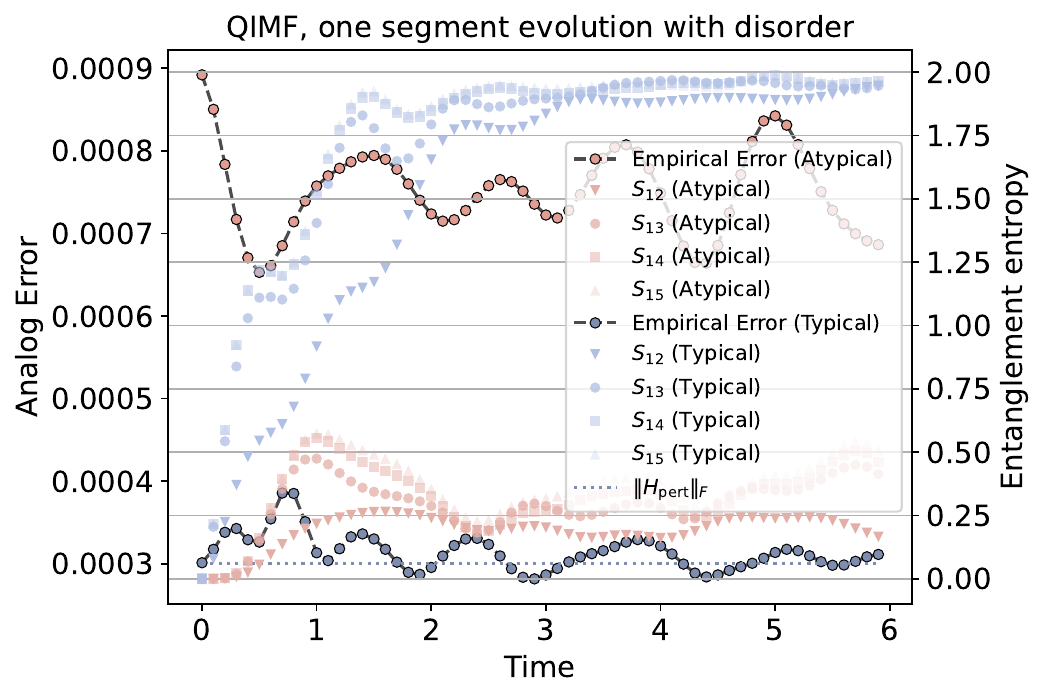}
    }
    
    \subfloat[]{
    \includegraphics[width=0.42\linewidth]{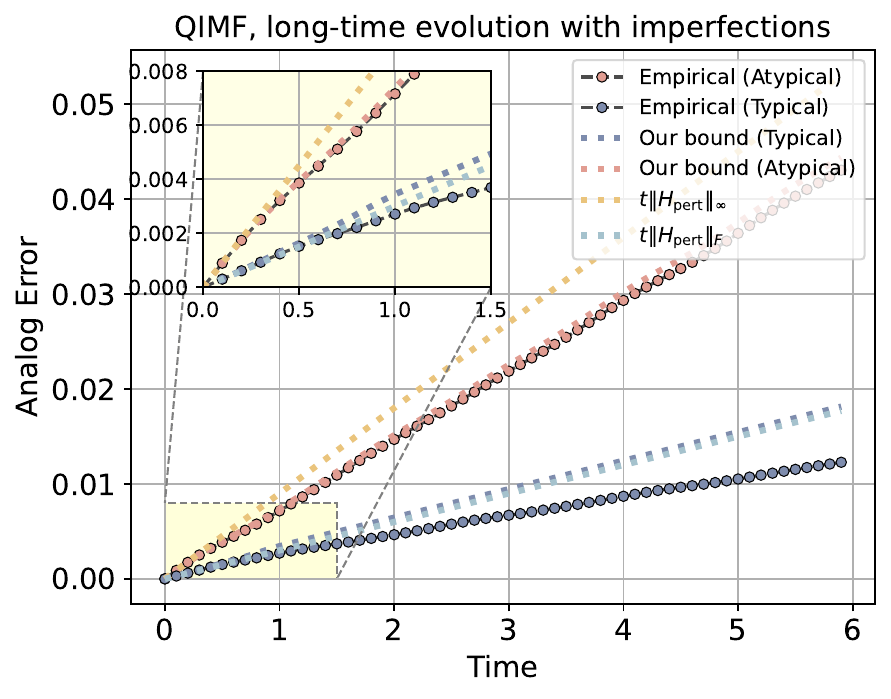}
    }
    \subfloat[]{
    \includegraphics[width=0.5\linewidth]{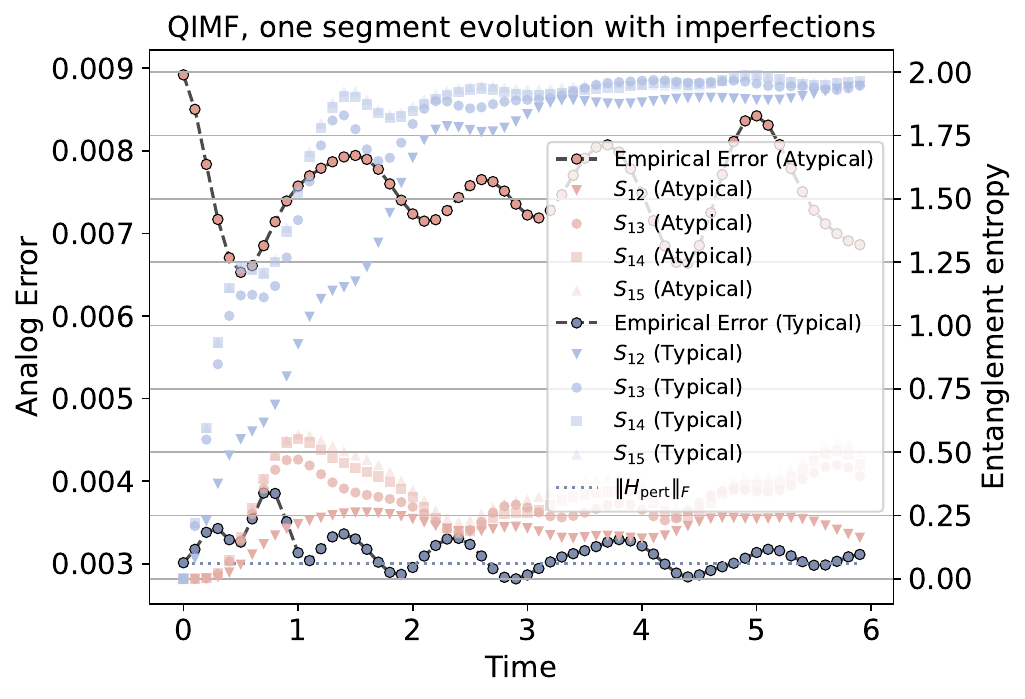}
    }

    \subfloat[]{
    \includegraphics[width=0.42\linewidth]{Figures/long-time_di,c=1_.pdf}
    }
    \subfloat[]{
    \includegraphics[width=0.5\linewidth]{Figures/one_segment_di_.pdf}
    }
    \caption{Error of Analog simulation in typical and atypical cases for 1D QIMF model. (a)(c)(e) Error of long-time analog simulation defined by the distance between real state and simulated state $\|(U_0(t)-U(t))\ket\psi\|$. Error models are with disorder only ($\delta_i\in\mathcal N(0,0.01),\eta=0$), with imperfection only ($\delta_i=0,\forall i$ and $\eta=0.01$) and with both disorder and imperfection terms ($\delta_i\in\mathcal N(0,0.01),\eta=0.01$), respectively. In typical input state, where the entanglement entropy grows during the system's evolution, our estimate closely matches the average-case performance, scaling with $\norm{H_{\text{pert}}}$. In contrast, in an atypical input state, where the entanglement entropy remains low throughout the evolution, the error scaling deviates significantly from the average case and approaches the worst-case bound characterized by the spectral norm. (b)(d)(f) Error of one-segment analog simulation with respect to different initial states defined as $
    \|(U_0(\delta t)-U(\delta t))\ket{\psi(t)}\|$.Error models are with disorder only, with imperfection only, and with both disorder and imperfection, respectively. The entanglement entropy of both typical and atypical cases defined on the subsystem of two qubits ($S_{1,2},~S_{1,3},~S_{1,4},~S_{1,5}$) is also shown.}
    \label{fig:placeholder}
\end{figure}

In practical experimental scenarios, the disorder term will include the effects of various perturbations, so it is more reasonable to treat the coefficient of the disorder term as a normally distributed random number $\delta_k\in\mathcal N(0,\sigma)$. In our numerical tests, the variance is chosen as $\sigma=0.001$. And we consider an analog Hamiltonian  
\begin{equation}
    H'=H_0+\sum_{k=1}^N\delta_kX_k+\eta\sum_{i=1}^{N-1}X_iX_{i+1},
\end{equation}
where $\eta=0.001$. In this setting, the analog simulation error is shown to be small for most cases.

\begin{figure}
    \centering
    \subfloat[]{
    \includegraphics[width=0.42\linewidth]{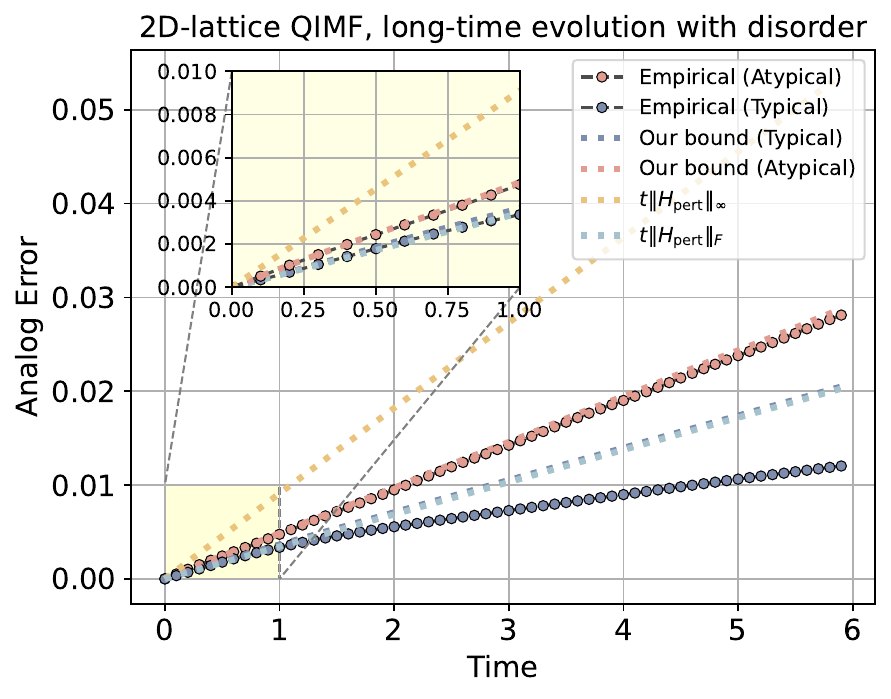}
    }
    \subfloat[]{
    \includegraphics[width=0.5\linewidth]{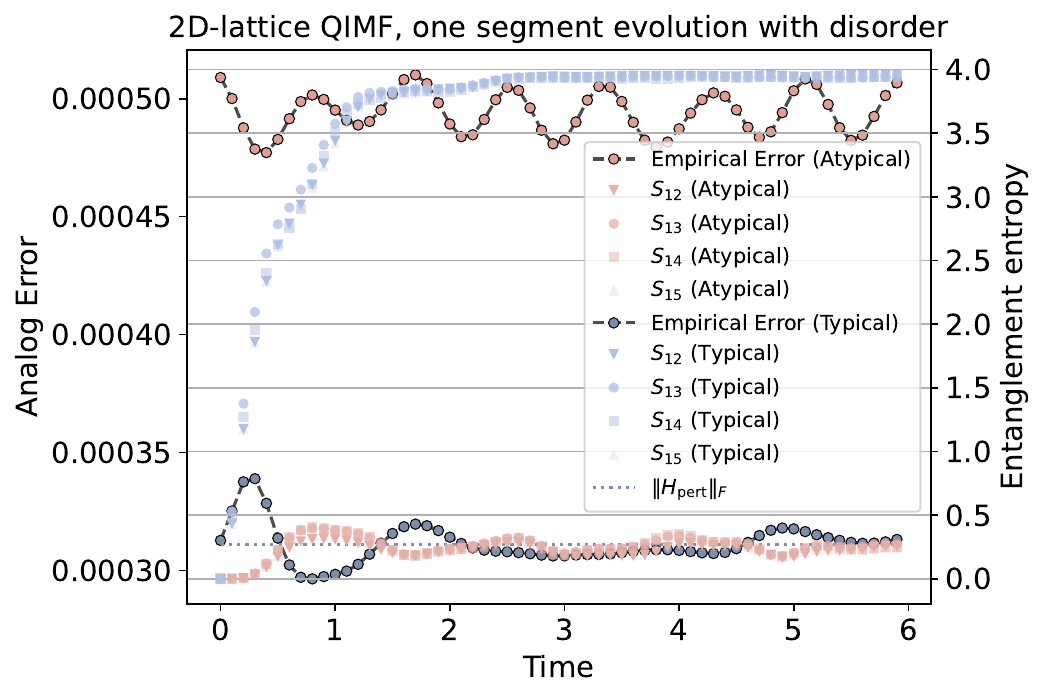}
    }
    
    \subfloat[]{
    \includegraphics[width=0.42\linewidth]{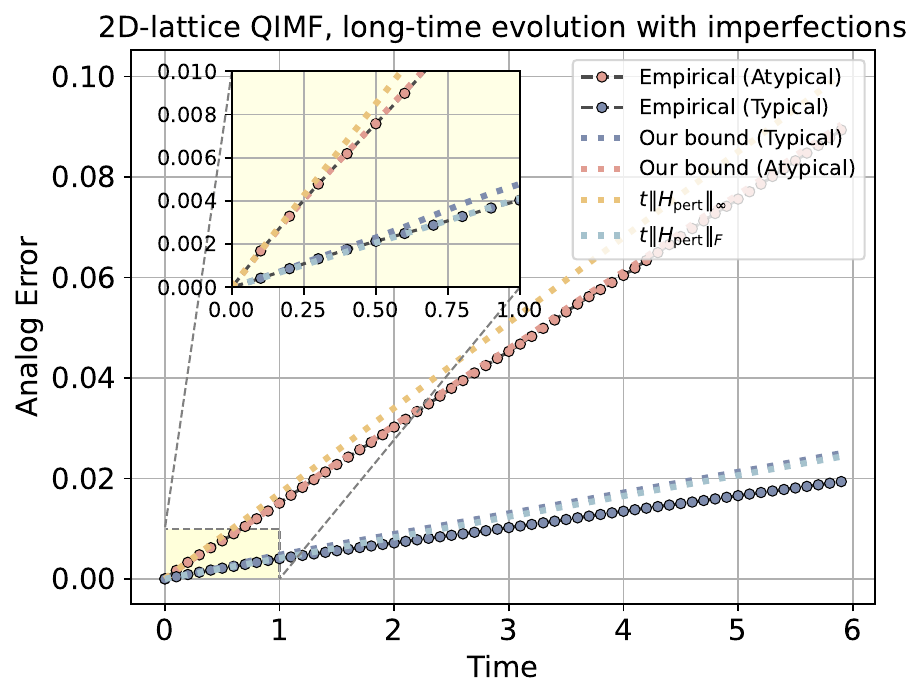}
    }
    \subfloat[]{
    \includegraphics[width=0.47\linewidth]{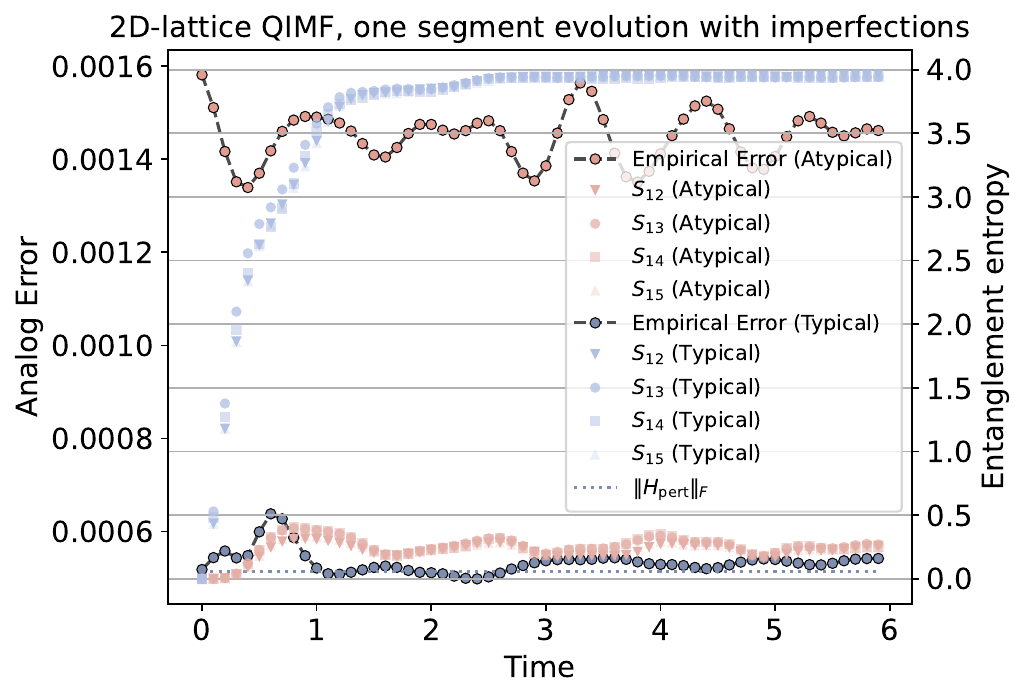}
    }

    \subfloat[]{
    \includegraphics[width=0.44\linewidth]{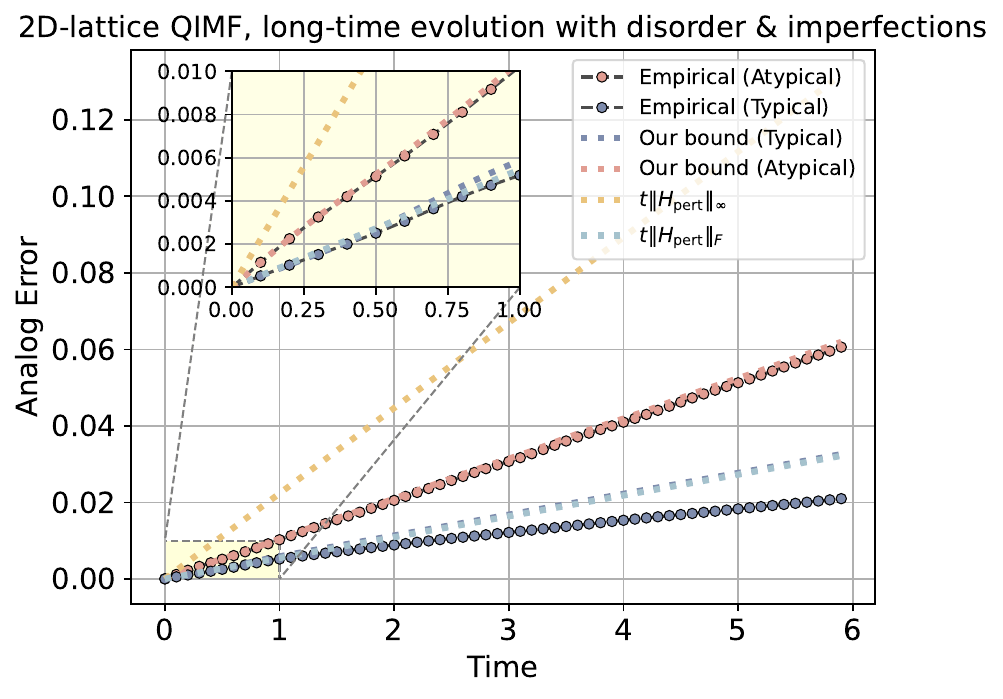}
    }
    \subfloat[]{
    \includegraphics[width=0.47\linewidth]{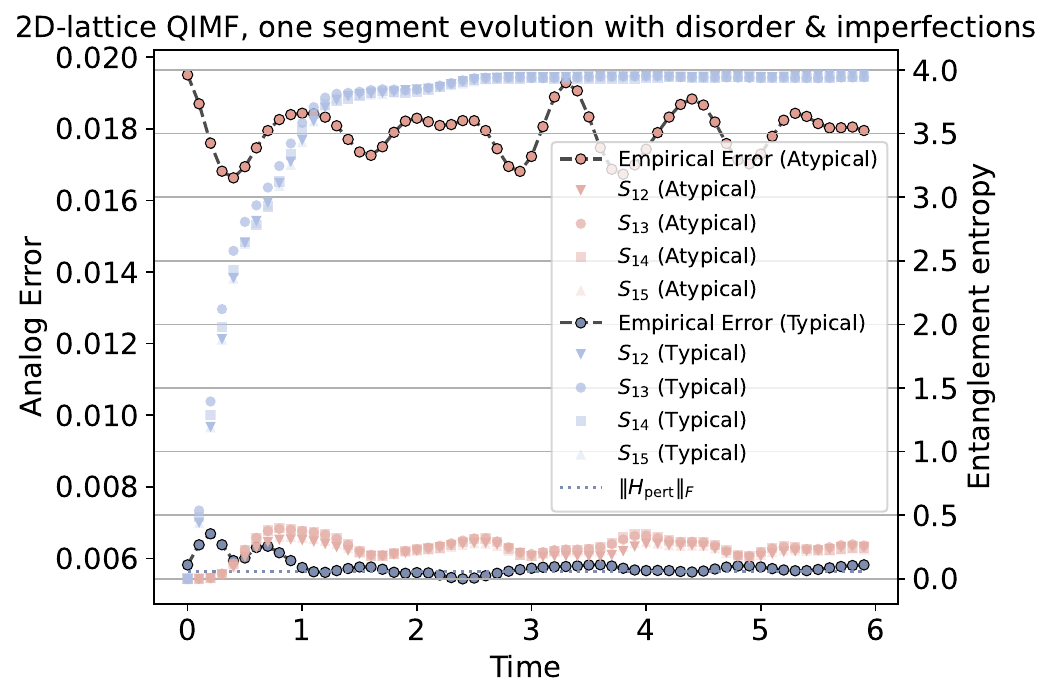}
    }
    \caption{Error of Analog simulation in typical and atypical cases for 2D QIMF model. (a)(c)(e) Error of long-time analog simulation defined by the distance between real state and simulated state $\|(U_0(t)-U(t))\ket\psi\|$. Error models are with disorder only ($\delta\in\mathcal N(0,0.001),\eta=0$), with imperfection only ($\delta=0,\eta=0.001$), and with both disorder and imperfection ($\delta=\mathcal N(0,0.001),\eta=0.001$), respectively. In typical input state, where the entanglement entropy grows during the system's evolution, our estimate closely matches the average-case performance, scaling with $\norm{H_{\text{pert}}}$. In contrast, in an atypical input state, where the entanglement entropy remains low throughout the evolution, the error scaling deviates significantly from the average case and approaches the worst-case bound characterized by the spectral norm. (b)(d)(f) Error of one-segment analog simulation with respect to different initial states defined as $
    \|(U_0(\delta t)-U(\delta t))\ket{\psi(t)}\|$.Error models are with disorder only, with imperfection only, and with both disorder and imperfection, respectively. The entanglement entropy of both typical and atypical cases defined on the subsystem of two qubits ($S_{1,2},~S_{1,3},~S_{1,4},~S_{1,5}$) is also shown.}
    \label{fig:placeholder2D}
\end{figure}

\emph{2D lattice QIMF model.---}
We also numerically test a QIMF model on a $4\times 3$ 2D spin lattice with Hamiltonian
\begin{equation}
    H_0=h_x\sum_i X_i+h_y\sum_iY_i+J\sum_{\langle i,j\rangle}X_iX_j,
\end{equation}
where $\langle
{i,j}\rangle$ represents an index pair of two nearest-neighbor spins, parameters $h_x=0.809$, $h_y=0.9045$, $J=1$. We also consider error models including stochastic disorder and control imperfection terms (see Figure~\ref{fig:placeholder2D}).

    


\subsection{{Fermi-Hubbard Model}}
The Fermi-Hubbard model is fundamental in condensed matter physics, used to describe interacting fermions on a lattice. The Hamiltonian of the 1D Fermi-Hubbard model 
\begin{equation}
        H_0=-\sum_{i=1}^L\sum_{\sigma}(a_{i,\sigma}^\dagger a_{i+1,\sigma}+\text{H.c.})+V\sum_{i=1}^La_{i,\uparrow}^\dagger a_{i,\uparrow}a_{i,\downarrow}^\dagger a_{i,\downarrow},
\end{equation}
captures the competition between the kinetic energy (with parameter $t$) of electrons and the on-site Coulomb repulsion (with parameter $V$). This simple yet rich model helps explain key phenomena such as magnetism, metal–insulator transitions, and high-temperature superconductivity. Fermionic systems are characterized by anticommutation relations between creation and annihilation operators, expressed as
\begin{equation}
    \{a_i,a_j^\dagger\}=\delta_{ij},\quad \{a_i,a_j\}=\{a_i^\dagger,a_j^\dagger\}=0.
\end{equation}
However, quantum computers operate on qubits that obey commutation rather than anticommutation relations. To simulate fermions on qubits, one uses transformations such as the Jordan–Wigner (JW) transform, which maps fermionic operators to strings of Pauli matrices:
\begin{equation}
    a_j=(\otimes_{k<j}Z_k)\otimes\dfrac{X_j+iY_j}2,\quad a_j^\dagger=(\otimes_{k<j}Z_k)\otimes\dfrac{X_j-iY_j}2.
\end{equation}
In our numerical tests, we consider simulating a Fermi-Hubbard model with parameters $t=1$ and $U=0.5$. Compared to the Hamiltonian $H_0$ of the target free fermion model, the Hamiltonian of the simulator $H_{\text{analog}}$ contains an additional noise term from the Coulomb potential.
\begin{align}
    H_0=&-\sum_{i=1}^{L-1}(c_j^{\dagger}c_{j+1}+c_{j+1}^\dagger c_j),\\
    H_{\text{pert}}=&\delta\sum_{i=1}^La^\dagger_{i, \uparrow} a_{i, \uparrow}
                         a^\dagger_{i, \downarrow} a_{i, \downarrow},\\
    H_{\text{analog}}=&H_0+H_{\text{pert}}.
\end{align}
%

\subsection{Entanglement-induced resilience in quantum control}

In the main text, we demonstrate the relationship between the effective error of a single-qubit gate with robust control pulses (RCPs) and entanglement entropy in a 2D lattice. The RCPs used in the main text is given by Ref.~\cite{hai2025scalablerobustquantumcontrol} and take the form
\begin{equation}
\Omega(t) = \Omega_m\sin(\pi t/T)(a_0+\sum_{
j=1}^n a_j \cos(2\pi jt/T+\phi_j)),
\end{equation} 
where parameters are listed in Table~\ref{Table:Pulse}.

\begin{table}[t]
    \centering

    \begin{tabular}{c c c c c}
        \hline
        RCPs & $\Omega_m$(MHz) & $T$(ns) & $a$ & $\phi$ \\
        \hline
        $X_\pi$ &
        150 & 180 &
        $[0.2374,\,0.2683,\,0.1459,\,0.0335,\,0.0030,\,0.0144]$ &
        $[-0.0055,\,-0.0021,\,-0.0006,\,-0.2457,\,-0.0157]$ \\

        $X_{\pi/2}$ &
        150 & 180 &
        $[0.1735,\,0.1438,\,0.0625,\,-0.0427,\,-0.0606,\,0.0207]$ &
        $[0.0013,\,0.0049,\,-0.0139,\,-0.0093,\,0.0062]$ \\

        $X_{2\pi}$ &
        150 & 180 &
        $[0.1522,\,0.1288,\,0.0434,\,-0.0866,\,-0.0375,\,-0.0174]$ &
        $[0.0093,\,-0.0431,\,0.0567,\,-0.0104,\,-0.0313]$ \\
        \hline
    \end{tabular}

    \caption{Optimized pulse parameters for different RCPs.}
    \label{Table:Pulse}
\end{table}

Here, we demonstrate a 2-qubit case where the implementation of a 2-qubit gate produces disturbances on 6 neighboring qubits. 
We consider the Hamiltonian given by Eq.~\eqref{eq:2qubit_rot} with $J=10\pi$MHz and perturbation Eq.~\eqref{eq:2qubit_per} with parameters $\delta =100$kHz, $J_\text{residue}=100$kHz. The numerical results are demonstrated in Figure~\ref{fig:twpqibitgate}. To illustrate the relationship between error and the entanglement entropy between subsystem of the 8 involved qubits and its complementary subsystem, we prepare a set of purified Gibbs states
\begin{equation}
    \ket{\psi_T}_{AB}=\dfrac{\sum_ie^{-\epsilon_i/T}\ket{\phi_i}_A\ket{\phi^*_i}_{B}}{\|\sum_ie^{-\epsilon_i/T}\ket{\phi_i}_A\ket{\phi^*_i}_{B}\|},\label{Apendix:gibs}
\end{equation}
where $A$ denotes the system composed of the evolved, the remaining system $B$ includes the other 8 qubits. $\{(\epsilon_i,\ket{\phi_i}_{A})\}$'s represent eigenvalues and eigenvectors of Hamiltonian $V=-\sum_{i\in A}Z_i$, and $\ket{\phi^*_i}_{B}$'s are corresponding states on $B$.


The reduced density matrix of the purified Gibbs state is 
\begin{equation}
    \rho_A=\tr_B(\ket{\psi_T}\bra{\psi_T})=\dfrac1{\mathcal Z}\sum_ie^{-2\epsilon_i/T}\ket{\phi_i}\bra{\phi_i},
\end{equation}
where $\mathcal Z=\sum_ie^{-2\epsilon_i/T}$. So the entanglement entropy defined by the von Neumann entropy is 
\begin{equation}
    S_A=-\tr(\rho_A\ln\rho_A)=\dfrac1{\mathcal Z}\sum_i\dfrac{2\epsilon_i}{T}e^{-2{\epsilon_i/T}}\ket{\phi_i}\bra{\phi_i}.
\end{equation}
And the gate error can be calculated by
\begin{equation}
    \|(U-U_0)\ket{\psi_T}\|=\tr [(U-U_0)^\dagger(U-U_0)\ket{\psi_T}\bra{\psi_T}]=\tr _A\{\tr_B[(U-U_0)^\dagger(U-U_0)]\rho_A\}
\end{equation}
We consider the implementation of a 2-qubit gate with perturbation Hamiltonian given by Eq.~\eqref{Eq_multidot_H_Pauli} with gate time $T=$25ns, 50ns, 75ns, 100ns, respectively.
\begin{figure}
    \centering
    \includegraphics[width=0.5\linewidth]{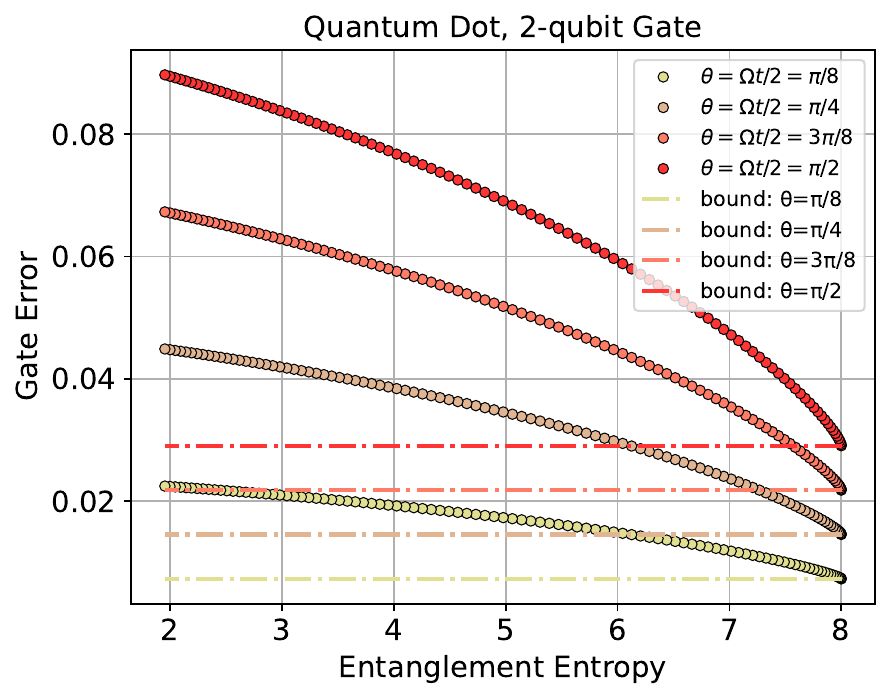}
    \caption{Implementing a 2-qubit gate, the gate error versus entanglement entropy. States are chosen according to Eq.~\eqref{Apendix:gibs}, where temperatures are chosen to be $T\in\left\{\frac1t|t\in\{0.008,0.016,\dots,0.48\}\right\}$. The entanglement entropy is calculated with respect to the subsystem of the first two qubits. When the entanglement entropy gets close to its maximum $\log d=2$, the gate error quickly reduces to the bound given by Eq~\eqref{eq:t-dependent_entangled}. The dependence of gate error on entropy is well described by a biquadratic root function, which is consistent with the entanglement-based bound with a maximally entangled state.}
    \label{fig:twpqibitgate}
\end{figure}

\section{Entanglement-induced resilience in quantum circuits}

Suppose a coherent noise model satisfying $U=U_0 e^{-i\lambda E},$ where $U_0$ is the ideal evolution and $e^{-i\lambda E}$ is the error unitary.
The $p$-order product formula of Trotterization is the special case of this model. Every single segment of the product formula can be modeled by the proposed model.
As we all known, a single (short-time) segment of product formula is given as $\mathscr{U}_p(\delta t)=U_0(\delta t) (I+\mathscr{M})$ where $\norm{\mathscr{M}}=\mc{O}(\delta t^{p+1})$. Since $I+\mathscr{M} $ is a unitary, one can denote it as $e^{\delta t^{p+1} M}=I+\delta t^{p+1}M+t^{2p+2}M^2/2...=I+\mathscr{M}$. Different from the Trotterization case, the model $U=U_0 e^{-i\lambda E}$ does not assume $U_0$ is a short-time evolution (close to identity), and it can be a highly nonlocal gate or strong coupling evolution.

In the following, $\norm{A}$ denote spectral norm, $\norm{A\ket{\psi}}=\sqrt{\bra{\psi}A^\dag A\ket{\psi}}$ denote vector norm. $S(\rho_A)=-\Tr(\rho_A\log\rho_A)$ is defined as the entanglement (Von Neumann) entropy. The entanglement-based lemma for a quantum circuit is given as follows.

\begin{lemma} [Entanglement-induced circuit resilience]Given a quantum state $\ket{\psi}$, an idea unitary/evolution $U_0$ and a approximate one $U$ with coherent error $e^{-i\lambda E}$ satisfying $U=U_0e^{-i\lambda E}$ ($E=\sum_j E_j$ is a Hermitian operator), the additive error of evolution is defined as $U-U_0$, and the evolved state is bounded by 

\begin{equation}
    \norm{(U-U_0)\ket{\psi}}= \mc{O}( \sqrt{\frac{\Tr(\mathscr{B})}{d}}+\sum_j \sqrt{{\norm{\mathscr{B}_j}} \sqrt{2\log(d_{\support(\mathscr{B}_{j})})-2S(\rho_{j})}}).
\end{equation}
 where 
$$\mathscr{B}=\sum_{s_1,s_2=1, s_1+s_2=\text{even}}^{s_1+s_2=K+1} \frac{(-1)^{s_2}({i}\lambda )^{s_1+s_2} (\sum_{j_1,j_2,..,j_{s1+s_2}}^J\prod_{i=1}^{s_1+s_2}E_{j_i})} {s_1!s_2!}, $$
and 
$\mathscr{B}_j=\mathscr{B}_{j(s_1,s_2)}=\frac{(-1)^{s_2}({i}\lambda )^{s_1+s_2} (\sum_{j_1,j_2,..,j_{s1+s_2}}^J\prod_{i=1}^{s_1+s_2}E_{j_i})} {s_1!s_2!}.$
\label{lemma:entanglement-induced}
\end{lemma}

\begin{proof}

Set
$\mathscr{A}=U-U_0=U_0(e^{-i\lambda E}-I)= U_0 {\sum_{s=1} \frac{(-i\lambda E)^s}{s!}}$,      
one has 
\begin{equation}
    \mathscr{A}^\dagger \mathscr{A}= {\sum_{s_1=1}^\infty \big[\frac{(-i\lambda E)^{s_1}}{s_1!}} \big]^\dagger  \big[{\sum_{s_2=1}^\infty  \frac{(-i\lambda E)^{s_2}}{s_2!}}\big].
\end{equation}
Since $e^{-i\lambda E}$ is an error unitary, and generally should be a perturbative rotation with some specific Hermitian operator $E$, where $E=\sum^J_j E_j$ and each $E_j$ is local (with constant weight or support). For example, a two-qubit gate or interaction evolution, its perturbation unitary is local because its driving Hamiltonian is local.  

Here $\lambda$ is a small parameter, we truncated the additive error up to $K$ order of $\lambda $, so that $$\mathscr{A}=U_0[{\sum^{K}_{s=1} \frac{(-i\lambda E)^s}{s!}}+\mc{O}(\lambda^{K+1})].$$
Now the positive semi-define  operator $\mathscr{A}^\dagger \mathscr{A}$ (leading term up to $K+1$ order) is given as

\begin{equation}
\begin{split}
    \mathscr{A}^\dagger\mathscr{A}&=[{\sum^{K}_{s_1=1} \frac{(-i\lambda E)^{s_1}}{s_1!}}+\mc{O}(\lambda^{K+1})]^\dagger [{\sum^{K}_{s_2=1} \frac{(-i\lambda E)^{s_2}}{s_2!}}+\mc{O}(\lambda^{K+1})]
    \\&=  \sum_{s_1,s_2=1}^{s_1+s_2=K+1} \frac{(-1)^{s_2}({i}\lambda E)^{s_1+s_2}}{s_1!s_2!} +\mc{O}(\lambda^{K+2})\\
     \\
    &=  \sum_{s_1,s_2=1, s_1+s_2=\text{even}}^{s_1+s_2=K+1} \frac{(-1)^{s_2}({i}\lambda E)^{s_1+s_2}}{s_1!s_2!} +\mc{O}(\lambda^{K+2})\\
    &= \sum_{s_1,s_2=1, s_1+s_2=\text{even}}^{s_1+s_2=K+1} \frac{(-1)^{s_2}({i}\lambda )^{s_1+s_2} (\sum_{j_1,j_2,..,j_{s1+s_2}}^J\prod_{i=1}^{s_1+s_2}E_{j_i})} {s_1!s_2!} +\mc{O}(\lambda^{K+2}).
    \end{split}
\end{equation}

For simplicity, set $\mathscr{A}^\dagger\mathscr{A}=\mathscr{B}+\mc{O}(\lambda^{K+2})=\sum_j \mathscr{B}_j+\mc{O}(\lambda^{K+2})$ where $\mathscr{B}_j=\mathscr{B}_{j(s_1,s_2)}=\frac{(-1)^{s_2}({i}\lambda )^{s_1+s_2} (\sum_{j_1,j_2,..,j_{s1+s_2}}^J\prod_{i=1}^{s_1+s_2}E_{j_i})} {s_1!s_2!}$.

Using the Lemma \autoref{lemma1} below, we have

\begin{equation}
    \norm{(U-U_0)\ket{\psi}}= \mc{O}(  \sqrt{\frac{\Tr(\mathscr{B})}{d}}+\sum_j \sqrt{{\norm{\mathscr{B}_j}} \sqrt{2\log(d_{\support(\mathscr{B}_{j})})-2S(\rho_{j})}}).
\end{equation}

\end{proof}

\begin{corollary}[1-order approximation]
Given approximate unitary $U=U_0 e^{-i\lambda E}$ where $U_0$ is the ideal one, for truncated order $K=1$, the additive error $\mathscr{A}=U-U_0=U_0(-i\lambda E+\mc{O} (\lambda^2))$ resulting $\mathscr{A}^\dagger \mathscr{A}=\lambda^2E^2+\mc{O}(\lambda^3)=\lambda^2 \sum_{j,j^\prime} E_jE_{j^\prime}+\mc{O}(\lambda^3)$. For any quatnum state $\ket{\psi}$, the additive error $\norm{\mathscr{A}\ket{\psi}}=\norm{(U-U_0)\ket{\psi}}$ is bounded by 

\begin{equation}
    \norm{(U-U_0)\ket{\psi}}= \mc{O}( \lambda \norm{E}_F+ \lambda\sum_j \sqrt{\norm{E_jE_{j^\prime}} \sqrt{2\log(d_{\support(E_{j}E_{j^\prime})})-2S(\rho_{j})}}),
\end{equation}
where  $\|A\|^2_F:= \Tr(A^{\dagger}A) /d$ is the (square) of the normalized Frobenius norm, $\rho_{j}:=\Tr_{[N]\setminus \support(A_j)}(\ket{\psi}\bra{\psi})$ is the reduced density matrix of $\ket{\psi}\bra{\psi}$ on the subsystem of $\support(A_j)$, and $\text{S}(\rho_{j})$ is the entanglement entropy of $\rho_{j}$.
\end{corollary}

\begin{corollary}[2-order approximation]
Given approximate unitary $U=U_0 e^{-i\lambda E}$ where $U_0$ is the ideal one, for truncated order $K=2$, the additive error $\mathscr{A}=U-U_0=U_0(-i\lambda E-\frac{\lambda^2E^2}{2}+\mc{O} (\lambda^3))$ resulting $$\mathscr{A}^\dagger \mathscr{A}=\lambda^2E^2-\frac{-i\lambda^3E^3+i\lambda^3E^3}{2}+\mc{O}(\lambda^4)=\lambda^2E^2+\mc{O}(\lambda^4).$$ For any quatnum state $\ket{\psi}$, the additive error $\norm{\mathscr{A}\ket{\psi}}=\norm{(U-U_0)\ket{\psi}}$ is bounded by 

\begin{equation}
    \norm{U-U_0\ket{\psi}}= \mc{O}( \lambda \norm{E}_F+\lambda\sum_{j,j^\prime} \sqrt{\norm{E_jE_{j^\prime}} \sqrt{2\log(d_{\support(E_{j}E_{j^\prime})})-2S(\rho_{j})}}),
\end{equation}
where  $\|A\|^2_F:= \Tr(A^{\dagger}A) /d$ is the (square) of the normalized Frobenius norm, $\rho_{j}:=\Tr_{[N]\setminus \support(A_j)}(\ket{\psi}\bra{\psi})$ is the reduced density matrix of $\ket{\psi}\bra{\psi}$ on the subsystem of $\support(A_j)$, and $\text{S}(\rho_{j})$ is the entanglement entropy of $\rho_{j}$.
\end{corollary}
 Note that $\Tr(E^2)/d=\norm{E}_F$ since $E^\dag=E$.
It can be seen that the additive error $\norm{\mathscr{A} \ket{\psi}}$ of 1-order and 2-order approximation of $e^{-i\lambda E}$ is the same while ignoring the
higher-order error terms. \\

Since the unitary here is not a short-time evolution, i.e., $U_0$ is far from $I$ and can be highly nonlocal, $\mathscr{A}=U_0(e^{-i\lambda E}-I)$ is not local anymore. Fortunately, $\mathscr{A}^\dagger \mathscr{A}$ can be local.  Now one can simply rewrite the Lemma proposed in \cite{Zhao_2025NP}, obtaining

\begin{lemma} \cite{Zhao_2025NP} \label{lemma1}
         Let $A=\sum_j A_j$ be a positive semi-defined operator, acting on $N$ qubits, where $A_j$ acts nontrivially on the subsystem with $\support(A_j)$. Then
\begin{align}
  | \bra{\psi}A \ket{\psi}|\le \frac{\Tr(A)}{d}+    \sum_{j} ~\| A_j\| \sqrt{2\log(d_{\support(A_{j})})-2S(\rho_{j})},
\end{align}
where $\rho_{j}:=\Tr_{[N]\setminus \support(A_j)}(\ket{\psi}\bra{\psi})$ is the reduced density matrix of $\ket{\psi}\bra{\psi}$ on the subsystem of $\support(A_j)$, and $\text{S}(\rho_{j})$ is the entanglement entropy of $\rho_{j}$.   
\end{lemma}

\begin{proof}
The term $A_j$ in the expression for $A$ only acts nontrivially on $\support(A_j)$. We denote its nontrivial part by $L_{j}:=\Tr_{[N]\setminus \support(A_j)} (A_j)$. Since $2^{-N}\Tr(A_j)=d_{\support(A_j)}^{-1}\Tr(L_{j})$, we have
\begin{equation}
\begin{aligned}
\sum_j\bra{\psi}A_j \ket{\psi}&=\sum_j\Tr(L_{j}\rho_{j})=\Tr[L_{j}(\rho_{j}- \mathbb{I}_{\support(A_j)}/d_{\support(A_j)})]+\sum_j\Tr(L_{j})/d_{\support(A_j)}\\
&= \sum_j\Tr[L_{j}(\rho_{j}- \mathbb{I}_{\support(A_j)}/d_{\support(A_j)})]+\Tr(A_j)/2^N\\
&\sum_j\leq \|L_{j}\|\Tr|\rho_{j}- \mathbb{I}_{\support(A_j)}/d_{\support(A_j)}|+\sum_j\Tr(A_j)/2^N\\
&=\sum_j \|A_j\| \Tr|\rho_{j}- \mathbb{I}_{\support(A_j)}/d_{\support(A_j)}|+\Tr(A)/2^N.
\end{aligned}
\end{equation}
Moreover, the trace distance of $\rho_{j}$ and $\mathbb{I}_{\support(A_j)}/d_{\support(A_j)}$ can be bounded by the relative entropy, the quantum Pinsker inequality, as
\begin{align}
    \Tr|\rho_{j}- \mathbb{I}_{\support(A_j)}/d_{\support(A_j)}| \le \sqrt{2S(\rho_{j}\|\mathbb{I}_{\support(A_j)}/d_{\support(A_j)}) }=
    \sqrt{2\log(d_{\support(A_j)})-2S(\rho_{j})}.
\end{align}
\end{proof}

If  $A=B^\dagger B$, where $B$ is a square matrix, Lemma \autoref{lemma1} reduces to the one in \cite{Zhao_2025NP}.

\end{document}